\documentclass[journal]{IEEEtran}
   \usepackage[pdftex]{graphicx}
\usepackage{amsmath}
\usepackage{algpseudocode,algorithm,algorithmicx}
\usepackage{hyperref}
\usepackage{amsthm}
\usepackage{amssymb}
\usepackage{cite}
\usepackage{subcaption}
\usepackage{mathtools}
\newcommand{\mykron}{\widetilde{\otimes}}
\def\r{r_q}
\def\n{i}
\def\m{j}

\def\q{q}
\def\z{x}
\def\Z{\tilde{H}}  
\def\ni{l}
\def\no{p}
\def\Zthreefour{Z}
\def\Zthreethree{\Psi}
\def\Ztwofour{\Phi} 
\def\Hperm{H}
\def\vecv{\nu}
\def\vecu{\rho}
\def\E{\mathcal{I}}
\def\S{s}
\def\C{c}

\newtheorem{lem}{Lemma}
\newtheorem{thm}{Theorem}

\usepackage{todonotes}

\begin{document}
%
\title{Structure-Preserving Model Reduction for Nonlinear Power Grid Network
\\[1ex]
\thanks{{This work was supported in parts by National Science Foundation under
Grant No. DMS-1923221}}}


\author{Bita~Safaee\thanks{B.~Safaee is with Department of Mechanical Engineering at Virginia Tech, Blacksburg, VA 24061, (e-mail: bsafaee@vt.edu) } and Serkan~Gugercin 
\thanks{S.~Gugercin is with Department of Mathematics at Virginia Tech, Blacksburg, VA 24061 (e-mail: gugercin@vt.edu)}}

\maketitle

\begin{abstract}
We develop a structure-preserving system-theoretic model reduction framework for nonlinear power grid networks. First, via a lifting transformation,
we convert the original nonlinear system with trigonometric nonlinearities to an equivalent quadratic nonlinear model.  This equivalent representation allows us to employ the $\mathcal{H}_2$-based model reduction approach, Quadratic Iterative Rational Krylov Algorithm (Q-IRKA),
as an intermediate model reduction step. 
Exploiting the structure of the underlying power network model,
we show that the model reduction bases resulting from Q-IRKA have a special subspace structure, which allows us to effectively construct the final model reduction basis. This final basis is applied on the original nonlinear structure to yield a reduced model that preserves the physically meaningful (second-order) structure of the original model.
The effectiveness of our proposed framework is illustrated via two numerical examples. 

 \end{abstract}


%
\IEEEpeerreviewmaketitle

\section{Introduction}
{P}{ower} networks are complex and large-scale systems in operation.
 Simulation of these high dimensional power system models are computationally expensive and demand unmanageable levels of storage in dynamic simulation or trajectory sensitivity analysis.
Hence, reducing the order of these models is of a great importance especially for real time applications, stability analysis, and control design \cite{Chow2013}.

Numerous model order reduction (MOR) techniques have been applied to the linearized power networks models including balanced truncation \cite{Liu2009,CheridBettayeb1991,RamirezKamalasadan2016,SturkSandberg2014,LeungVillella2019}, Krylov subspace (interpolatory) methods \cite{ChaniotisPai2005,WangYu2013,WangDing2015,safaeeGugercin2021}, proper orthogonal decomposition (POD) \cite{WangPai2014}, singular perturbation theory \cite{PaiAdgaonkar1981,ChowSauer1990}, 
sparse representation of the system matrices \cite{LevronBelikov2017}, and clustering analysis \cite{ChengJ2018,Schaft2014,IshizakiAihara2015,mlinaric2015efficient}. For a comparative study of MOR techniques in power systems, see, e.g., \cite{BESSELINKSchilders2013, AnnakkageMartinez2012,Chow2013}.
Recently, development of MOR techniques for nonlinear systems including power networks has received great attention. 
In \cite{ParriloMarsden1999}, POD is employed to reduce the hybrid, nonlinear model of a power network for the study of cascading failures.
Trajectory piece-wise linearization (TPWL) method to approximate the nonlinear term in swing models together with POD to find a reduced order model for the swing dynamics was employed in \cite{MalikDiez2016}.
Balanced truncation based on empirical controllability and observability covariances of nonlinear power systems is proposed in \cite{QiDimitrovski2017}.
The empirical Gramians and balanced realization were also used in
\cite{ZhaoRen2017,ZhaoShi2013}
to  build a nonlinear power system model  suitable for controller design.
In \cite{PurvineWang2017}, real-time phasor measurements are used to cluster generators based on similar behaviors using recursive spectral bipartitioning and spectral clustering.
The reduced model is formed by representative generators chosen from each cluster.
A model reduction approach that adaptively switches between a hybrid system with selective linearization and a linear system based on the variations of the system state or size of a disturbance was developed in \cite{OsipovSun2018}. 
A structure-preserving, aggregation-based model order reduction framework based on synchronization properties of power systems
is introduced in \cite{MlinaricIshizaki2018}.
In \cite{WangZhou2012} an SVD-based approach for real-time dynamic model reduction with the ability of preserving  some nonlinearity in the reduced model is presented.
The recent work \cite{TobiasGrundel2020} presents the idea of transforming the swing equations into a quadratic system, lifts nonlinear swing equations to a quadratic one, and employs quadratic balanced truncation model order reduction technique.  A non-intrusive data-driven modeling framework  for  power  network  dynamics  using  Lift  and Learn \cite{QianWillcox2020} has been studied in \cite{safaeeGugercin22021}. 

In this paper, we focus on designing a structure-preserving and input-independent MOR framework for a special class of power grid networks. We obtain the reduction bases using the idea of lifting the nonlinear dynamics to quadratic models and employing $\mathcal{H}_2$-quasi-optimal model reduction for quadratic systems (Q-IRKA) \cite{BennerGoyalGugercin2018}. The reduction basis for the second-order nonlinear power network grids model is obtained with the help of the reduction bases obtained by Q-IRKA and their particular subspace structure.  This reduction basis is then applied on the second-order nonlinear power network grids model to form a structure-preserving reduced model.

The rest of the paper is organized as follows: Section \ref{sec:swing_model} briefly introduces the nonlinear dynamical model of the power grid network followed by a quick recap of projection-based
structure-preserving model reduction in this setting in 
Section \ref{sec:SO_MOR}. Section \ref{sec:quadratic_rep} illustrates how to lift the second-order dynamics of the power grid network to quadratic dynamics and presents an equilibrium analysis. In Section \ref{sec:MOR_Quad}, we describe the model reduction for quadratic systems via the  $\mathcal{H}_2$-based model reduction, Q-IRKA, followed by the required modification for the quadratized power grid model before employing  $\mathcal{H}_2$-based model reduction. In Section \ref{sec:proposed_approach}, we present our structured-preserving model reduction approach for power grid networks by proving and exploiting the specific structure arising in Q-IRKA for these dynamics. Section \ref{sec:numerical_result} illustrates the
feasibility of our approach via numerical examples and by comparing our approach with two other methods, followed
by conclusions in Section \ref{sec:conclusion}.

\section{power grid network model}\label{sec:swing_model}
{There are three most commonly used models for describing a dynamical model of a power grid network}: \textit{synchronous motor} ($SM$), \textit{effective network} ($EN$) and \textit{structure-preserving} ($SP$). 
Each model is described as a network of {$n$} coupled {oscillators (generators or loads)} whose dynamics at the {$i$th node} is expressed by 
the nonlinear second-order equation
\cite{NishikawaMotter2015}
\begin{align}\label{eq:oscillators}
\frac{2J_\n}{\omega_R}\ddot{\delta_\n}+\frac{D_\n}{\omega_R}\dot{\delta_\n} + \sum_{\substack{\m =1 , \m \ne \n}}^n K_{\n\m}\sin (\delta_\n - \delta_\m- \gamma_{\n\m})= B_\n, 
\end{align}
where $\delta_\n$ is the angle of rotation for the $\n$th oscillator; $J_\n$ and $D_\n$ are inertia and damping constants, respectively; $\omega_R$ is the angular frequency for the system; $K_{\n \m} \geq 0$ is dynamical coupling between oscillator $\n$ and $\m$; and $\gamma_{\n \m}$ is the phase shift in this coupling. Constants $B_\n$, $K_{\n \m}$ and $\gamma_{\n \m}$ are computed by solving the power flow equations and applying Kron reduction \cite{NishikawaMotter2015}.
{The three models mainly differ by the way they model  loads. The EN model assumes loads as constant impedances instead of oscillators. The SP model represents loads as first-order oscillators ($J_i=0$) while the SM model represent  loads as synchronous motors expressed as second-order oscillators.}

Define the state-vector $\delta = [\delta_1~\delta_2~\dots ~ \delta_n]^T \in \mathbb{R}^{n}$. 
Then, the dynamics of a network of $n$ coupled oscillators, given in
\eqref{eq:oscillators} for the $i$th node, can be described in the system form 
\begin{align}\label{eq:So_Dyn} 
     \mathcal{M} \ddot{\delta}(t) + \mathcal{D} \dot{\delta}(t) + f(\delta) &= \mathcal{B} u(t),\\
     y(t) &= \mathcal{C}\delta(t), \label{eq:So_Out}
\end{align}
where the matrices $\mathcal{M}$, $\mathcal{D} {\in \mathbb{R}}^{n  \times n}$, and $\mathcal{B}\in \mathbb{R}^{n}$  are defined as
\begin{align}
\begin{array}{rcl}
      \mathcal{M}   & \hspace{-1ex} = \hspace{-1ex} & \mathrm{diag}(m_1,\dots,m_n),~~  
     \mathcal{D}  = \mathrm{diag}(d_1,\dots,d_n),~  \\
   \mathcal{B}  & \hspace{-1ex} = \hspace{-1ex} & \begin{bmatrix}B_1, & B_2,&\ldots,& B_n
    \end{bmatrix}^T,
    \end{array}
\end{align}
where $m_i = \frac{2J_i}{\omega_R}$ and $d_i= \frac{D_{i}}{\omega_R}$, {$u(t) =1$} is input to the dynamics;
and the nonlinear map $f: \mathbb{R}^{n} \to \mathbb{R}^{n}$ is defined such that its $i$th component for $i=1,2,\ldots,n$ is 
\begin{align} \label{eq:fs}
    \left(f(\delta)\right)_i = \sum_{\substack{\m =1 , \m \ne \n}}^n K_{\n\m}\sin (\delta_\n - \delta_\m- \gamma_{\n\m}).
\end{align}
The variable $y(t) \in \mathbb{R}^\no$ in \eqref{eq:So_Out} corresponds to the output (quantity of interest) of the network dynamics,  described by the state-to-output matrix  $\mathcal{C} \in  \mathbb{R}^{\no \times n}$. 

\section{Structure-preserving MOR approach} \label{sec:SO_MOR}
Large scale power network dynamics, i.e., those with large $n$,
are computationally expensive to simulate and control. With the aid of the model reduction, we are able to replace these high dimensional models by lower dimensional models that can closely mimic the input-output behavior of the original high dimensional systems. 
To be more specific, we seek to develop a nonlinear structure-preserving model reduction approach such that the reduced model preserves the physically-meaningful second-order structure. In other words,
we directly reduce the second-order dynamics \eqref{eq:So_Dyn}-\eqref{eq:So_Out} to obtain a reduced second-order system with dimension $r \ll n $ and of the form
\begin{align} \label{eq:So_ROM}
     \mathcal{M}_r \ddot{\delta_r}(t) + \mathcal{D}_r \dot{\delta_r}(t) + f_r(\delta_r) & =  \mathcal{B}_r u(t),\\
     y_r(t) & =  \mathcal{C}_r\delta_r(t), \label{eq:So_ROMOut}
\end{align}
where $\delta_r \in \mathbb{R}^r$ is the reduced state;  $\mathcal{M}_r,\mathcal{D}_r \in \mathbb{R}^{r \times r}$;  $\mathcal{B}_r \in \mathbb{R}^{r \times 1}$; $f_r: \mathbb{R}^r \to \mathbb{R}^r$; and $\mathcal{C}_r \in \mathbb{R}^{\no \times r}$. The goal of this structure-preserving model reduction process is that $y_r(t)$ approximates the original output $ y(t)$ with high fidelity.

Given the full-order dynamics
\eqref{eq:So_Dyn}-\eqref{eq:So_Out}, the most common approach for constructing the reduced model \eqref{eq:So_ROM}-\eqref{eq:So_ROMOut}
is the Petrov-Galerkin projection framework: Construct two model reduction bases $\mathcal{W},\mathcal{V} \in \mathbb{R}^{n \times r}$  such that $\delta(t) \approx \mathcal{V} \delta_r(t)$. Insert  $\mathcal{V} \delta_r(t)$ into \eqref{eq:So_Dyn} and enforce a Petrov-Galerkin condition on the residual to obtain the reduced matrices in \eqref{eq:So_ROM} and \eqref{eq:So_ROMOut} as
\begin{align} \label{eq:So_ROM_matrices}
\begin{array}{rcl}
  \mathcal{M}_r & \hspace{-1ex} = \hspace{-1ex} & \mathcal{W}^T \mathcal{M} \mathcal{V}, \ \mathcal{D}_r = \mathcal{W}^T \mathcal{D} \mathcal{V},  \\ 
f_r(\delta_r)  & \hspace{-1ex} = \hspace{-1ex} & \mathcal{W}^T f(\mathcal{V}\delta_r), \ \mathcal{B}_r = \mathcal{W}^T \mathcal{B} ,\ \mbox{and}\  \mathcal{C}_r = \mathcal{C} \mathcal{V}.
\end{array}
\end{align}
To preserve the symmetry and positive definiteness of $\mathcal{M}$ and $\mathcal{D}$ in \eqref{eq:So_ROM}, we employ a Galerkin projection by setting $\mathcal{W} = \mathcal{V}$.  
However, unlike the usual MOR via Galerkin projection  where the subspace $\mathcal{V}$ (and thus $\mathcal{W}$) only contains information from the dynamics equation 
\eqref{eq:So_Dyn} and ignores the output equation 
\eqref{eq:So_Out}, our one-sided Galerkin projection will contain information from both \eqref{eq:So_Dyn} and \eqref{eq:So_Out} by taking advantage of the special structure arising in the network dynamics. These issues are discussed in detail in Section \ref{sec:proposed_approach}. 
The  term $f_r(\delta_r)$ in
\eqref{eq:So_ROM_matrices} is usually further approximated to resolve the lifting bottleneck \cite{ChaturantabutSorensen2010}. Due to the page limit, we skip this detail here.

The accuracy of the structure-preserving reduced model \eqref{eq:So_ROM} depends on the proper choice of $\mathcal{V}$.
For general nonlinearities, POD is usually the method of choice.
However, for special classes of nonlinear systems such as bilinear and quadratic-bilinear systems, one can indeed use well-established system theoretical methods; see, e.g., \cite{BennerGoyalGugercin2018, BennerGoyal2017, BennerBreiten2012_2}.
A large class of nonlinear systems,
such as those with smooth nonlinearities, e.g.,  exponential, trigonometric, can indeed be represented as quadratic-bilinear systems by introducing new variables (in a lifting map) \cite{mccormick1976computability}, \cite{Gu2011}, \cite{BennerTobias2015}, \cite{TobiasGrundel2020}, {\cite{QianWillcox2020}}.
Converting a nonlinear model to its equivalent quadratic form provides an opportunity to perform input-independent model reduction techniques where a system-theoretic norm is defined \cite{BennerGoyal2017,BennerGoyalGugercin2018}. Although this transformation is not unique, it is exact in many cases (as in here), i.e., there is no approximation error {\cite{BennerGoyal2017}}. 

Second-order model \eqref{eq:oscillators} inherently contains a quadratic nonlinearity in $f(\delta)$ and  the nonlinear system \eqref{eq:So_Dyn} can be lifted to a quadratic dynamic as we will see next in Section \ref{sec:quadratic_rep}.
 However, as stated in section \eqref{sec:SO_MOR}, we seek to obtain a reduced second-order model, not a reduced quadratic model. Therefore, by exploiting the inherent structure of the model reduction bases obtained for the quadratic model, we form the final reduction base $\mathcal{V}$ to construct the final structure 
 preserving reduced model \eqref{eq:So_ROM_matrices}. This is achieved in Sections \ref{sec:MOR_Quad}--\ref{sec:proposed_approach}.
\section{Quadratic representation}\label{sec:quadratic_rep}
{As we mentioned in Section \ref{sec:SO_MOR} for systems with quadratic nonlinearities,  effective and in some cases optimal, systems-theoretic MOR methods already exist. Therefore,
in obtaining high-fidelity reduced models,
we can significantly benefit from a quadratic representation of the nonlinear swing dynamics.}

Towards this goal, recall the original second-order dynamics \eqref{eq:oscillators}. {Expand the nonlinearity $\sin (\delta_\n - \delta_\m- \gamma_{\n\m})$ using the trigonometric identity} and introducing $s:= \sin$ and $c := \cos$
\begin{align}\label{eq:trigonometric}
     \S(\delta_\n -  \delta_\m- \gamma_{\n\m})  & =   (\S(\delta_\n)\C(\delta_\m)-\C(\delta_\n)\S(\delta_\m))\C(\gamma_{\n\m})\\
   & -   (\C(\delta_\n)\C(\delta_\m)+\S(\delta_\n)\S(\delta_\m))\S(\gamma_{\n\m}). \notag
\end{align}
Note that equation \eqref{eq:trigonometric} is quadratic in $\sin(\delta_i)$ and $\cos(\delta_i)$. This directly hints at choosing $\sin(\delta)$ and $\cos(\delta)$ as auxiliary variables in a new state vector to convert the nonlinear dynamics to a quadratic one. This is precisely what has been done in  \cite{TobiasGrundel2020} to find a quadratic representation  
for \eqref{eq:So_Dyn}. However, since the transformation, the so-called lifting, is not unique and we have a slightly different form in the resulting state-space form, we include its derivation here for completeness. 

\begin{lem}
Given the second-order dynamics \eqref{eq:So_Dyn}, define the new lifted state vector $\q(t) \in \mathbb{R}^{4n}$ as
\begin{align}\label{eq:vector_x}
  \q(t)&= \begin{bmatrix}
 \delta(t)^T & \dot{\delta}(t)^T & \sin(\delta(t))^T & \cos(\delta(t))^T  \end{bmatrix}^T \notag\\
 &= \begin{bmatrix}
q_1^T&q_2^T&q_3^T&q_4^T
\end{bmatrix}^T.
\end{align}
For the new state $q(t)$ in \eqref{eq:vector_x}, the original dynamics  \eqref{eq:So_Dyn} can be written exactly as a quadratic nonlinear system of the form
\begin{align} \label{eq:permuted_quad}
\begin{array}{rcl}
 E\dot \q(t) & \hspace{-1ex} = \hspace{-1ex} & A\q(t) + \Hperm(\q(t) \otimes \q(t)) + Bu(t) \\
 y(t) & \hspace{-1ex} = \hspace{-1ex} & C\q(t), 
 \end{array}
\end{align}
where $E, A$ $\in \mathbb{R}^{4n \times 4n}$, $B \in \mathbb{R}^{{4n \times 1}}$ , $C \in \mathbb{R}^{\no \times 4n}$,  $H \in \mathbb{R}^{4n \times (4n)^2}$, and $\otimes$ denotes the Kronecker product.
 \end{lem}
\begin{proof}
The proof is constructive, i.e., it will show how the matrices 
$E,A,H,B$ and $C$ in \eqref{eq:permuted_quad} will be constructed. 
Using the new state vector $q$ in \eqref{eq:vector_x}, we can directly rewrite the original output $y(t) = \mathcal{C} \delta(t)$ as
$y(t) = C q(t)$ as in \eqref{eq:permuted_quad}
with
\begin{align}\label{eq:c_q}
    C =\begin{bmatrix}
          \mathcal{C} & 0 & 0 & 0
    \end{bmatrix} \in \mathbb{R}^{\no \times 4n}.
\end{align}
It follows from \eqref{eq:So_Dyn} and \eqref{eq:vector_x}  that
{
\begin{align} \label{eq:dotq1}
    \dot{q}_1 & = q_2 , \\ 
   \mathcal{M} \dot{q}_2 & = - \mathcal{D} q_2 - f(q_1) + \mathcal{B} u ,\label{eq:dotq2} \\ 
   \dot{q}_3 & = q_2 \circ q_4 \label{eq:dotq3} , \\ 
   \dot{q}_4 & = - q_2 \circ q_3 , \label{eq:dotq4}
\end{align}
where $\circ$ denotes the Hadamard product}.

We start by writing $\dot{q_3}$ in \eqref{eq:dotq3} and $\dot{q}_4$ in \eqref{eq:dotq4} as
\begin{align}\label{eq:dotq3_dotq4}
\dot{q}_3 =  \Phi (q_2 \otimes q_4)\quad \mbox{and} \quad
\dot{q}_4 =  -\Phi (q_2 \otimes q_3),
\end{align}
where  $\Ztwofour = \begin{bmatrix}
    \Ztwofour_1 & \Ztwofour_2 & \dots & \Ztwofour_n
    \end{bmatrix} \in \mathbb{R}^{n \times n^2}$ and
  $\Ztwofour_k \in \mathbb{R}^{n \times n}, \text{for}~ k=1,\ldots,n,$ 
 has zero entries except for the  $k$th diagonal element, which is $1$. 
Next, we write  the nonlinearity  $f(q_1)$  in \eqref{eq:dotq2}  as a quadratic term.
{Towards this goal, we use 
\eqref{eq:trigonometric} in \eqref{eq:fs} to obtain
{\small
\begin{align}
       (f (q_1))_i =  \notag  &\sum_{\substack{\m =1 , \m \ne \n}}^n K_{\n\m}\left( (q_{3})_i (q_{4})_j - (q_{4})_i (q_{3})_j \right) \cos (\gamma_{\n\m}) \\
    & - \sum_{\substack{\m =1,  \m \ne \n}}^n K_{\n\m}((q_4)_i (q_{4})_j)\sin (\gamma_{ij}) \label{eq:fs_q} \\
    & - \sum_{\substack{\m =1 \notag , \m \ne \n}}^n K_{\n\m}\left( (q_{3})_i (q_{3})_j \right)\sin(\gamma_{\n\m}).
\end{align}
}}
Using \eqref{eq:fs_q}, we can write the vector $f(q_1)$ compactly as
\begin{align}\label{eq:fs_q_quad}
  f (q_1) = -\left( \Zthreefour(\q_{3} \otimes \q_{4}) + \Zthreethree(\q_{4} \otimes \q_{4})+ \Zthreethree(\q_{3} \otimes \q_{3}) \right) ,
\end{align}
where $\Zthreefour = \begin{bmatrix}
    \Zthreefour_1 & \Zthreefour_2 & \dots & \Zthreefour_n
    \end{bmatrix} \in \mathbb{R}^{n \times n^2}$
    consists of $n$ block matrices $Z_k \in \mathbb{R}^{n \times n}$ defined as
{\small
\begin{align}\label{eq:Zk}
    & \Zthreefour_k(i,j) = \Bigg\{ \begin{matrix}
    K_{kj}\cos{(\gamma_{kj})}, &  i=j=1,...,n , \ j \ne k\\ -K_{kj}\cos{(\gamma_{kj})}, & i = k, j=1,...,n , \ j \ne k\\
    0 & \text{otherwise}
    \end{matrix}.
    \end{align}
}Similarly, we have $\Zthreethree = \begin{bmatrix}
    \Zthreethree_1 & \Zthreethree_2 & \dots & \Zthreethree_n
    \end{bmatrix}\in \mathbb{R}^{n \times n^2}$
and for the $k^{th}$ block of $\Psi$, denoted by $\Psi_k \in \mathbb{R}^{n \times n}$, we have
{\small
\begin{align} \label{eq:Psik}
    & \Zthreethree_k(i,j) = \Bigg\{ \begin{matrix}
    -\frac{1}{2}{K_{kj}}\sin{(\gamma_{kj})}, &  i=j=1,...,n , \ j \ne k\\ -\frac{1}{2}{K_{kj}}\sin{(\gamma_{kj})} ,& i = k, j=1,...,n , \ j \ne k\\
    0 & \text{otherwise}
    \end{matrix}.
    \end{align}
}
{Using \eqref{eq:dotq3_dotq4} and \eqref{eq:fs_q_quad}, we  now rewrite \eqref{eq:dotq1}-\eqref{eq:dotq4} as}
\begin{align} \label{eq:dotq_rewriteq1}
    \dot{q}_1  & =  q_2 ,  \\ 
  \mathcal{M} \dot{q}_2 & =  - \mathcal{D} q_2 + \Zthreefour(\q_{3} \otimes \q_{4}) ,\notag
  \\ &~\qquad\qquad +\Zthreethree(\q_{4} \otimes \q_{4})+ \Zthreethree(\q_{3} \otimes \q_{3}) \label{eq:dotq_rewriteq2} + \mathcal{B} u, \\
  \dot{q}_3 & =  \Phi (q_2 \otimes q_4), \label{eq:dotq_rewriteq3} \\ 
  \dot{q}_4 & =   -  \Phi (q_2 \otimes q_3). \label{eq:dotq_rewriteq4}
\end{align}
Note that the new formulation of the dynamics in \eqref{eq:dotq_rewriteq1}-\eqref{eq:dotq_rewriteq4}
only contains linear and quadratic terms in $q$, and 
the input mapping is linear. The output mapping $C$ is already established in \eqref{eq:c_q}.
Therefore, 
there exist matrices $E,A,H,B$ and $C$ such that the original power network dynamics can be written as a quadratic nonlinear system as  in \eqref{eq:permuted_quad}, thus proving the desired result. However, to give a  constructive proof, we continue to  show how the matrices in \eqref{eq:permuted_quad} can be explicitly constructed.

It immediately follows from \eqref{eq:dotq_rewriteq1}-\eqref{eq:dotq_rewriteq4} that $E \in \mathbb{R}^{4n \times 4n} $, $A \in \mathbb{R}^{4n \times 4n} $, and $B \in \mathbb{R}^{4n}$ in \eqref{eq:permuted_quad} (corresponding to the linear parts of the dynamics) are given by
 \begin{align}\label{eq:b_q}
E = &~\text{blkdiag}(I,\mathcal{M},I,I),  \\
\label{eq:a_q}
\quad B & =
    \begin{bmatrix}
    0 \\ \mathcal{B} \\ 0 \\ 0
    \end{bmatrix}, \mbox{and}\quad
  A =  \begin{bmatrix} 0 & I & 0 & 0\\ 0 & -\mathcal{D} & 0 & 0\\0 & 0  & 0 & 0\\0 & 0 & 0 & 0\end{bmatrix},
 \end{align}
where $I$ denotes the identity matrix (of appropriate size).
To show how $\Hperm$ is constructed, we 
define a revised  Kronecker product as
\begin{align}\label{eq:ourKron}
\q \, \mykron \, \q = [\q_{1} \, \mykron \, \q;\ \q_{2}\, \mykron\, \q;\ \q_{3}\, \mykron\, \q;\ \q_{4}\, \mykron\, \q] ,
\end{align}
where
$\q_{i}\, \mykron\, \q = \begin{bmatrix}
\q_{i}\otimes \q_{1};\q_{i}\otimes \q_{2};\q_{i} \otimes \q_{3}; \q_{i} \otimes \q_{4}
\end{bmatrix}$, for $i=1,\dots,4$.
Clearly, \eqref{eq:ourKron} is  a permuted form of the regular Kronecker product. It is introduced here to  make the derivation of $\Hperm$ easier.  
The quadratic terms \eqref{eq:dotq_rewriteq2}-\eqref{eq:dotq_rewriteq4}  can be written as
\begin{align}
    \Z (q\,\mykron \,q)\quad\mbox{where}\quad\Z \in \mathbb{R}^{4n \times (4n)^2}.
\end{align}
Decompose $\Z$ into four sub-matrices:
\begin{align}\label{eq:H_tiled}
    \Z = \begin{bmatrix}
    \Z_1 & \Z_2 & \Z_3 & \Z_4
    \end{bmatrix},
\end{align}
{where $\Z_k \in \mathbb{R}^{4n \times 4n^2}$ for $k=1,2,3,4$. The first submatrix $\Z_1$ corresponds to the first block of \eqref{eq:ourKron}, i.e, $
q_1 \tilde{\otimes}\, q = \begin{bmatrix}\q_{1}\otimes \q_{1};\q_{1}\otimes \q_{2};\q_{1} \otimes \q_{3}; \q_{1} \otimes \q_{4}
\end{bmatrix}$. There is no $\q_{1} \otimes \q_{i}$ term in \eqref{eq:dotq_rewriteq1}-\eqref{eq:dotq_rewriteq4}. Thus, we set $\Z_1 = 0$.

Similarly, the matrix $\Z_2$ corresponds to the second block of \eqref{eq:ourKron}, i.e, $q_2 \tilde{\otimes} q$. We can see from \eqref{eq:dotq_rewriteq1}-\eqref{eq:dotq_rewriteq4} that two terms, namely $\q_{2} \otimes \q_{3}$ and $ \q_{2} \otimes \q_{4}$, exist in \eqref{eq:dotq_rewriteq3} and \eqref{eq:dotq_rewriteq4}, respectively. One can rewrite \eqref{eq:dotq_rewriteq3} as $\dot{q}_3  =  \frac{\Phi}{2} (q_2 \otimes q_4) + \frac{\Phi}{2} (q_4 \otimes q_2)$ and allocate  
the first term, i.e., $\frac{\Phi}{2} (q_2 \otimes q_4)$ 
in $\Z (q\,\mykron\,q)$, to $\Z_2$. More specifically and by using the MATLAB notation for row and column indices, we set 
the $\Z_2(2n+1:3n,3n^2+1:4n^2) = \frac{\Phi}{2}$.  (The second term  in $\dot{q}_3$, i.e.,$\frac{\Phi}{2} (q_4 \otimes q_2)$,  will be matched by $\Z_4$ below). Similarly, we can write \eqref{eq:dotq_rewriteq4} as $\dot{q}_4  =  -\frac{\Phi}{2} (q_2 \otimes q_3) - \frac{\Phi}{2} (q_3 \otimes q_2)$ and allocate the first term to  $\Z_2$ by setting
$\Z_2(3n+1:4n,2n^2+1:3n^2) = -\frac{\Phi}{2}$ (the second term will be matched via $\Z_3$). Thus, we set 
{\small \begin{align}\label{eq:H2}
    \Z_2 = \begin{bmatrix}
    0 & 0 & 0 & 0\\
    0 & 0 & 0 & 0\\
    0 & 0 & 0 & \frac{\Ztwofour}{2}\\
    0 & 0 & \frac{-\Ztwofour}{2} & 0
    \end{bmatrix}.
\end{align}}Following similar arguments, we obtain}
{\small
\begin{align}\label{eq:H}
& \Z_3 = \begin{bmatrix}
    0 & 0 & 0 & 0\\
    0 & 0 & {\Zthreethree} & \frac{\Zthreefour}{2}\\
    0 & 0 & 0 & 0\\
    0 & \frac{-\Ztwofour}{2} & 0 & 0
    \end{bmatrix},~\Z_4 = \begin{bmatrix}
    0 & 0 & 0 & 0\\
    0 & 0 & \frac{-\Zthreefour}{2} & {\Zthreethree}\\
    0 & \frac{\Ztwofour}{2} & 0 & 0\\
    0 & 0 & 0 & 0
        \end{bmatrix}.
\end{align}
}
{Since \eqref{eq:ourKron} is a permuted form of $(\q \otimes \q)$, there exists a permutation matrix
$P_q \in \mathbb{R}^{n^2 \times n^2}$ such that}
  \vspace{-.2cm}
 \begin{align}\label{eq:kronecker_rel}
       (\q \, \mykron \, \q) =  P^T(\q \otimes \q),
 \end{align} 
where
 $P = \text{blkdiag}\left(P_q,P_q,P_q,P_q\right)$. Then, $\Hperm$ in \eqref{eq:permuted_quad} is
 \begin{align} \label{eq:finalHperm}
     \Hperm = \Z P^T.
 \end{align}
\end{proof}
\subsection{Equilibrium analysis}
{So far we have represented the power network dynamics in two formats: as second-order dynamics in \eqref{eq:So_Dyn}-\eqref{eq:So_Out} and as quadratic dynamics in \eqref{eq:permuted_quad}. In this section, we  investigate how the lifting transformation \eqref{eq:vector_x} affects the equilibrium points.}

To find the equilibrium points of \eqref{eq:So_Dyn}, first we transform  it to a first-order form by defining a state vector {$\mathcal{X}(t) = [\mathcal{X}_1(t)^T ~ \mathcal{X}_2(t)^T]^T = [\delta(t)^T ~ \dot{\delta}(t)^T]^T \in \mathbb{R}^{2n}$}  to obtain
{\small
\begin{equation}\label{eq:first_rewrite}
 \dot{\mathcal{X}}(t)=
  \begin{bmatrix}
0 & I\\0 & -\mathcal{M}^{-1}\mathcal{D}
\end{bmatrix}\mathcal{X}(t) + 
  \begin{bmatrix}
0\\-\mathcal{M}^{-1}f(\mathcal{X}_1) 
\end{bmatrix}
 + \begin{bmatrix}
 0\\ \mathcal{M}^{-1}\mathcal{B} \end{bmatrix},
\end{equation}
}

\noindent using
$u(t)=1$. Let $
    \mathcal{X}^* = \begin{bmatrix}
    {\mathcal{X}_1^*}^T & {\mathcal{X}_2^*}^T
    \end{bmatrix}^T$ be an equilibrium point of \eqref{eq:first_rewrite}.
By setting $\dot{\mathcal{X}}(t)=0$, we  obtain that  $\mathcal{X}^*$ satisfies
\begin{align}\label{eq:Eq_So} 
    \mathcal{X}^* = \begin{bmatrix}
    \mathcal{X}_1^* \\ \mathcal{X}_2^*
    \end{bmatrix} = \begin{bmatrix}
    \delta^* \\ 0
    \end{bmatrix}~\mbox{where}~~f(\delta^*) = \mathcal{B}.
\end{align}
\begin{lem}
Let $\mathcal{X}^* = \begin{bmatrix}
    {\delta^*}^T & 0^T
    \end{bmatrix}^T$ be an equilibrium point of 
 the original dynamics \eqref{eq:first_rewrite}. Then,
 \begin{align}\label{eq:EQ_quad}
    q^* = \begin{bmatrix}
    {\delta^*}^T & {0}^T & {(\sin{(\delta^*)})}^T & {(\cos{(\delta^*)})}^T
    \end{bmatrix}^T 
\end{align}
is an equilibrium point of \eqref{eq:permuted_quad}.
Similarly, let 
\begin{align}\label{eq:EQ_quad2}
    q^* & =  \begin{bmatrix}
    {q_1^{*}}^T& {q_2^*}^T& {q_3^*}^T& {q_4^*}^T
    \end{bmatrix}^T 
\end{align} 
be an equilibrium point of \eqref{eq:permuted_quad}. Then, $\q_2^*=0$
and 
$\mathcal{X}^* = \begin{bmatrix}
    {\q_1^*}^T & 0^T
    \end{bmatrix}^T$
    is an equilibrium point of \eqref{eq:first_rewrite}.
\end{lem}
\begin{proof}
Rewrite \eqref{eq:first_rewrite} as
\begin{align}\label{eq:general_firstorder}
    \dot{\mathcal{X}} = \mathcal{F}(\mathcal{X}),
\end{align}
where $\mathcal{F}: \mathbb{R}^{2n} \to \mathbb{R}^{2n}$. If $\mathcal{X}^*$ in \eqref{eq:Eq_So} is an equilibrium of \eqref{eq:general_firstorder}, then we have
    $\mathcal{F}(\mathcal{X}^*)=0$.
Now, {let us write \eqref{eq:vector_x}} as
\begin{align}\label{eq:q_lifting_map}
    q = \mathcal{T}(\mathcal{X}),
\end{align}
where $\mathcal{T}$ is the quadratic lifting in \eqref{eq:vector_x}, given by
\begin{align} \label{eq:lifting_map}
 \mathcal{T} : \mathcal{X} = \left(  \begin{matrix}
\mathcal{X}_1 \\ \mathcal{X}_2
 \end{matrix}  \right) \rightarrow q =\left(  \begin{matrix}
\mathcal{X}_1 \\ \mathcal{X}_2 \\ \sin{(\mathcal{X}_1)}\\ \cos{(\mathcal{X}_1)}
 \end{matrix}  \right) .
\end{align}
Differentiate \eqref{eq:q_lifting_map} to obtain
\begin{align}\label{eq:jacobian_eq}
    \dot{q}   = J_{\mathcal{T}}\dot{\mathcal{X}}
  =  J_{\mathcal{T}}\mathcal{F}(\mathcal{X}):=\mathcal{G}(q),
\end{align}
where the Jacobian $J_{{T}}$ is given by
\begin{align}\label{eq:jacobian}
    J_{\mathcal{T}} = \begin{bmatrix}
    I & 0\\
    0 & I\\
    \text{diag}(\cos{(\mathcal{X}_1)}) & 0\\
    -\text{diag}(\sin{(\mathcal{X}_1)}) & 0
    \end{bmatrix} \in \mathbb{R}^{4n\times 2n}.
\end{align}
Insert 
$\mathcal{F}(\mathcal{X}^*)=0$
into \eqref{eq:jacobian_eq} to obtain
$     \dot{q} = \mathcal{G}(q^*) =  J_{\mathcal{T}}\mathcal{F}(\mathcal{X}^*)=0,$
proving  that \eqref{eq:EQ_quad} is an equilibrium of \eqref{eq:permuted_quad}.

Now let $q^*$ in \eqref{eq:EQ_quad2} be an equilibrium of \eqref{eq:permuted_quad}. Then from \eqref{eq:jacobian_eq}, we obtain
$     \mathcal{G}(q^*)= J_{\mathcal{T}}\mathcal{F}(\mathcal{X}^*) = 0.$
Since the Jacobian \eqref{eq:jacobian} has full column rank, we have
$    \mathcal{F}(\mathcal{X}^*) = 0,$
{which by \eqref{eq:general_firstorder} yields that \eqref{eq:Eq_So} is an equilibrium of \eqref{eq:first_rewrite}.}
{The fact that $q_2^*=0$  follows from the second row block of $\mathcal{F}(\mathcal{X}^*) = 0$.}
\end{proof}

 \section{Model reduction for quadratic systems}\label{sec:MOR_Quad}
Now that we have represented the original nonlinear swing dynamics as an equivalent quadratic dynamics, we will investigate the MOR problem for those systems in more details. 

Consider the quadratic dynamical system 
\begin{equation} \label{eq:general_quadratic_original}
\Sigma:= \begin{cases}
$$ {E}\dot \z(t) = {A}\z(t) + H(\z(t) \otimes \z(t)) + {B} {u}(t) $$\\
$$~~ y(t) = C\z(t), ~~~ \z(0) = 0$$ 
\end{cases},
\end{equation}
where $E,A \in \mathbb{R}^{N \times N}$, $H \in \mathbb{R}^{N \times N^2}$, $B \in \mathbb{R}^{N \times \ni}$, and $C \in \mathbb{R}^{\no \times N}$. 
One can view $\Hperm$  as the mode-1 matricization 
of a third-order tensor $\mathcal{H} \in \mathbb{R}^{N \times N \times N} $. In other words, let  $\mathcal{\Hperm}_i \in \mathbb{R}^{N \times N}$ for $ i= 1,\dots,N$ be the frontal slices of $\mathcal{H}$. Then,
\begin{align}\label{eq:mode1_matricization}
      \Hperm =  {\mathcal{\Hperm}^{(1)} } &= \begin{bmatrix}
    \mathcal{\Hperm}_1 & \mathcal{\Hperm}_2 & \dots & \mathcal{\Hperm}_{N}
    \end{bmatrix},
\end{align}
where the notation ${\mathcal{\Hperm}^{(1)} }$ denotes the mode-1 matricization.

For the cases where the state-space dimension $N$ is  large,
the goal is to find a reduced quadratic system 
\begin{equation} \label{eq:general_quadratic_reduced}
\Sigma_r:= \begin{cases}
$$ {E}_r\dot{{\z}_r}(t) = {{A}_r}{\z}_r(t) + {H}_r ({\z}_r(t) \otimes {\z}_r(t)) + {{B}}_r {u}(t) $$\\
$$~~~ {y}_r(t) = {C}_r{\z}_r(t), ~~~ {\z}_r(0) = 0$$
\end{cases},
\end{equation}
where ${E}_{r},{{A}}_r \in \mathbb{R}^{\r \times \r}$
${H}_r \in \mathbb{R}^{\r \times \r^2}$, ${C}_r \in \mathbb{R}^{\no \times \r} $ and ${{B}}_r \in \mathbb{R}^{\r \times \ni}$ with $\r \ll N$ such that $y(t) \approx y_r(t)$ for a wide range of inputs.
We construct the reduce system \eqref{eq:general_quadratic_reduced} via projection: Construct two model reduction bases $V$, $W \in \mathbb{R}^{N \times \r}$  such that the reduced matrices in \eqref{eq:general_quadratic_reduced} are given by
\begin{align}\label{eq:E_QuadraticReducedMatrices}
& {E}_r = W^T E V,~ {{A}}_r = W^T {{A}} V , ~ {H}_r = W ^T H (V \otimes V), \notag\\
& {{B}}_r = W^T {B}, ~ C_r = C V. 
\end{align}
The quality of the reduced system \eqref{eq:general_quadratic_reduced} depends on the  choice of the model reduction bases $V$ and $W$. There are various model reduction techniques to determine theses bases including trajectory-based method such as POD (see, e.g., \cite{HinzeVolkwein2005}); and input-independent systems theoretic methods for quadratic systems such as balanced truncation  \cite{BennerGoyal2017}, \cite{TobiasGrundel2020}, interpolatory projections \cite{BennerTobias2015}, \cite{Gu2011}  
and $\mathcal{H}_2$-quasi-optimal model reduction \cite{BennerGoyalGugercin2018}. In this paper, we employ the $\mathcal{H}_2$-based approach in  \cite{BennerGoyalGugercin2018}. 
\subsection{$\mathcal{H}_2$-based model reduction of quadratic nonlinearities} \label{sec:H2_MOR}
Consider the original quadratic dynamics \eqref{eq:general_quadratic_original} with an asymptotically  stable matrix $E^{-1}{A}$.
Define a truncated $\mathcal{H}_2$ norm for \eqref{eq:general_quadratic_original} denoted by $\|\Sigma\|_{\mathcal{H}_2}(\tau)$ using the three leading Volterra kernels as \cite{BennerGoyalGugercin2018}:
\begin{align} \label{eq:h2_norm}
			& \| \Sigma \|^2_{{\mathcal{H}}_2^{(\mathcal T)}} := \\ \notag
			& \text{tr}\left( {\sum_{i =1 }^3\int_0^\infty \cdots \int_0^\infty h_i(t_1,\ldots,t_i) h_i^T(t_1,\ldots,t_i)dt_1\cdots dt_i}\right),
		\end{align}
where $h_1(t_1) = C e^{E^{-1}At_1}E^{-1}B$,
$h_2(t_1,t_2)  = 0$,  
{\small
$$
h_3(t_1,t_2,  t_3)  = C e^{E^{-1}At_3}E^{-1}H(e^{E^{-1}At_2}E^{-1}B \otimes e^{E^{-1}At_1}E^{-1}B)
$$}

\vspace{-3ex}
\noindent and {$\text{tr}(\cdot)$ denotes the trace of the argument (for details see, e.g., \cite{BennerGoyalGugercin2018,AntBG20}).} Note that when $H =0$, only the first kernel $h_1(t_1)$ remains and accordingly \eqref{eq:h2_norm} boils down to the
classical $\mathcal{H}_2$-norm of an asymptotically stable linear system. 

The goal is now to construct $V$ and $W$ such that the truncated $\mathcal{H}_2$ error norm $\|\Sigma - \Sigma_r\|_{\mathcal{H}_2}(\tau)$ is minimized.
This can be achieved by the model reduction algorithm Q-IRKA \cite{BennerGoyalGugercin2018}. A brief sketch of  (the two-sided) Q-IRKA is shown in Algorithm \ref{alg:Q-IRKA}. In some settings, to retain the positive definiteness and/or symmetry of the original realization in the reduced order matrices \eqref{eq:E_QuadraticReducedMatrices}, a one sided version of Q-IRKA
is applied where $W$ is set to $V$ as given in Algorithm \ref{alg:Q-IRKA_onesided}. The proposed approach to structure preserving model reduction of power networks as outlined in the next section will employ Q-IRKA as an intermediate step. Note that the reduced models from Q-IRKA is  quadratic, do not preserve the original second-order structure  \eqref{eq:So_Dyn}, and thus  cannot be directly used in a structure preserving setting. In the next section, we will prove  that the model reduction basis $V$ and $W$ resulting from applying Q-IRKA  to the quadratic form of power network have special subspace structure, which will then allow us to develop our structure-preserving model reduction algorithm. 

\begin{algorithm}
   \caption{ Two-sided Q-IRKA}\label{alg:Q-IRKA}
    \begin{algorithmic}[1]
     \Statex \textbf{Input}: $E$,\ ${A}$,\ $H$,\ ${B}$,\ $C$ 
     \Statex \textbf{Output}: $E_r$,\ ${A}_{{r}}$,\ ${H}_{r}$,\ ${{B}}_r$,\ $C_r$, and  $V$, $W$
     \State Symmetrize $H$, then transform it to a $N \times N \times N$ tensor and determine its mode-2 matricization $\mathcal{H}^{(2)}$ (for details see, e.g., \cite{Kolda2009}, \cite{BennerGoyalGugercin2018})
          \State  Make an initial guess for  $E_r$,\ ${A}_{{r}}$,\ ${H}_{r}$,\ ${{B}}_r$,\ $C_r$ 
    \While{not converged}
          \State Perform the spectral decomposition of the pair $E_r$ and ${A}_{{r}}$, i.e.,
          ${A}_{{r}}R = E_r R \Lambda $ and define:
          \Statex $\hat{H}=(E_r R)^{-1} {H}_{r}(R \otimes R)$~, $\hat{B}= (E_r R)^{-1} {B}_r$,$~\hat{C}=C_r R$
          \State Compute mode-2 matricization $\hat{\mathcal{H}}^{(2)}$
          \State Solve for $V_1$ and $V_2$: 
          \Statex $-E V_1 \Lambda - {A} V_1 =  {B} \hat{B}^T$
          \Statex $-E V_2 \Lambda - {A}V_2 = H (V_1 \otimes V_1)\hat{H}^T$ 
         \State Solve for $W_1$ and $W_2$:
          \Statex $-{E}^{T} W_1 \Lambda - {A}^T W_1 = C^T\hat{C}$
          \Statex $-{E}^{T} W_2 \Lambda - {A}^T W_2 = \mathcal{H}^{(2)} (V_1 \otimes W_1) (\hat{\mathcal{H}}^{(2)})^T$
          \State Compute $V\coloneqq V_1 + V_2$ and $W\coloneqq W_1 + W_2$.
           \State Create an orthogonal basis for each $V$ and $W$
           \State Determine the reduced matrices as \eqref{eq:E_QuadraticReducedMatrices}.
    \EndWhile
   \end{algorithmic}
\end{algorithm}

\begin{algorithm}
   \caption{ One-sided Q-IRKA}\label{alg:Q-IRKA_onesided}
    \begin{algorithmic}[1]
     \Statex \textbf{Input}: $E$,\ ${A}$,\ $H$,\ ${B}$,\ $C$ 
     \Statex \textbf{Output}: $E_r$,\ ${A}_{{r}}$,\ ${H}_{r}$,\ ${{B}}_r$,\ $C_r$, and  $V$, $W$
     \State Apply Algorithm \ref{alg:Q-IRKA} by replacing Step 7  with 
     $$W_1 = V_1 ~~\text{and}~~ W_2 = V_2 $$
   \end{algorithmic}
\end{algorithm}
\subsection{Modifications to perform Q-IRKA for power network dynamics: zero initial conditions and asymptotic stability} 

{Recall the state vector $\q(t)$ in \eqref{eq:vector_x}. Since $\sin(\delta)$ and $\cos(\delta)$ can not be simultaneously zero, the
reformulation \eqref{eq:permuted_quad} as a quadratic model will always have a nonzero initial condition. 
Since the most system theoretic methods to model reduction assumes a zero initial condition,  a  natural approach is to consider a shifted state vector that allows the use of the reduction methods developed for zero initial condition. 
Following \cite{TobiasGrundel2020}, we define a shifted state vector $\z = \q - \q_0$, such that in the new state $\z$, we  have $\z(0) = 0$} and  write 
\begin{align} \label{eq:xKronx}
     \q \otimes \q  = \z  \otimes \z + \left( (I \otimes \q_0) + ( \q_0 \otimes I) \right)\z + \q_0 \otimes \q_0. 
\end{align}
By substituting \eqref{eq:xKronx} in \eqref{eq:permuted_quad} we obtain
\begin{align} \label{eq:final_quad}
\begin{array}{rcl}
 {E}\dot \z(t) & \hspace{-1ex} = \hspace{-1ex} & \tilde{A}\z(t) + \Hperm(\z(t) \otimes \z(t)) + \tilde{B} \tilde{u}(t) \\
 y(t) & \hspace{-1ex} = \hspace{-1ex} & C\z(t), ~~~ \z(0) = 0, 
 \end{array}
\end{align}
\begin{align} \label{eq:etild}
\mbox{with}~~~\Tilde{A} & = A + H \left((I \otimes \q_0) + (\q_0 \otimes I)\right), ~   \tilde{u}  = \begin{bmatrix}
 u & 1
 \end{bmatrix}^T, \notag \\
  \tilde{B} & = [B ~~ A \q_0 + H \ (\q_0 \otimes \q_0)].
\end{align}
{Given the dynamical system \eqref{eq:etild} with zero initial conditions, one can apply a wide range of available methods, as mentioned at the beginning of Section \ref{sec:MOR_Quad}, to construct a quadratic reduced model. However, as in the linear case,
optimal-$\mathcal{H}_2$ model reduction for quadratic system requires asymptotic
stability of $E^{-1}{A}$ 
\cite{BennerGoyalGugercin2018}.}
Recall the quadratic dynamics \eqref{eq:final_quad} and $\tilde{A}$ in \eqref{eq:etild}.
Consider zero initial conditions for $\delta$ and $\dot \delta$, and thus the initial condition {$\q_0 =[0_{1 \times n} ~ 0_{1 \times n} ~ 0_{1 \times n} ~ 1_{1 \times n}]^{T} \in \mathbb{R}^{4n}$}. Thus, $\tilde{A}$ have  zero eigenvalues and therefore $E^{-1}\tilde{A}$ is not asymptotically stable \cite{TobiasGrundel2020}. To overcome this stability issue in the power network setting, we replace $\tilde{A}$ in \eqref{eq:final_quad} by 
\begin{align} \label{eq:Atildemu}
 \tilde{A_{\mu}} = \tilde{A} - \mu E ,
\end{align}
 where $\mu >0 $ is real and small, so that all the eigenvalues of the pair $\tilde{A_{\mu}}$ and $E$ have negative real part. Therefore, as the final modification, we replace {\eqref{eq:final_quad}} by
\begin{align} \label{eq:final_quad_MU}
\begin{array}{rcl}
 {E}\dot \z(t) & \hspace{-1ex} = \hspace{-1ex} & \tilde{A_{\mu}}\z(t) + \Hperm(\z(t) \otimes \z(t)) + \tilde{B} \tilde{u}(t) \\ 
 y(t) & \hspace{-1ex} = \hspace{-1ex}  & C\z(t), ~~~ \z(0) = 0.
 \end{array}
\end{align}
The quadratic dynamical system in \eqref{eq:final_quad_MU} has now asymptotically stable linear part and have zero initial conditions. Therefore, it has the structure to perform
$\mathcal{H}_2$ based model reduction using  Q-IRKA. 
However, we re-emphasize that the reduced-model obtained from applying Q-IRKA to \eqref{eq:final_quad_MU} is not structure-preserving; it does not inherit the second-order structure. Therefore, we employ Q-IRKA only to obtain potential reduction subspace information. 
{In the next section where we present the proposed framework, we describe how we  benefit from and how we use the subspaces resulting from Q-IRKA to obtain a reduced second-order model in \eqref{eq:So_ROM}. We also show that Q-IRKA subspaces applied to \eqref{eq:final_quad_MU} has special properties due to the underlying power network structure.}
\section{The proposed approach for MOR of Power Networks}\label{sec:proposed_approach}
Our aim is to perform structure-preserving model reduction of the second-order dynamics \eqref{eq:So_Dyn}-\eqref{eq:So_Out} using the reduction bases $\mathcal{V},\mathcal{W} \in \mathbb{R}^{n \times r}$  to obtain the reduced second-order model \eqref{eq:So_ROM}-\eqref{eq:So_ROM_matrices}. To keep the symmetry of the underlying dynamics, we perform a Galerkin projection by choosing $\mathcal{W} = \mathcal{V}$. Converting the nonlinear second-order dynamics to the quadratic form has allowed us to use (optimal) systems-theoretic techniques to apply, such as Q-IRKA. 
Although  we cannot simply use the output of Q-IRKA as the reduced model (since it does not preserve the second-order structure), 
we still would like to  employ the  high-fidelity model reduction bases $V$ and $W$ resulting from Q-IRKA in constructing $\mathcal{V}$. 
In this section, by analyzing and exploiting the structure of the underlying  power network dynamics, we show how to efficiently construct   $\mathcal{V}$ using $V$ and $W$ from Q-IRKA.
\subsection{Structures arising in Q-IRKA in the Power Network Setting}

Let $V \in \mathbb{R}^{4n \times \r}$ and $W \in \mathbb{R}^{4n \times \r} $ be the model reduction bases obtained from applying Q-IRKA to \eqref{eq:final_quad_MU}. Let 
$V_T \in \mathbb{R}^{n \times \r}$ and 
$W_T \in \mathbb{R}^{n \times \r}$ denote leading $n$ rows of $V$ and $W$.  Recall that due to \eqref{eq:vector_x}, 
the leading $n$ rows of $q(t)$, the state of quadratic dynamic \eqref{eq:final_quad_MU}, corresponds to the original (second-order) state $\delta (t)$. Therefore, one can  consider choosing 
$\mathcal{V} = V_T$ and $\mathcal{W} = W_T$ in \eqref{eq:So_ROM_matrices}. However, in addition to this being a Petrov-Galerkin projection and thus not preserving the symmetry, we will show in the next result $W_T$ has a particular subspace structure that one can further exploit in constructing the structured reduced model.
\begin{thm}\label{rankoneW}
Let $V$ and $W$ be obtained via Q-IRKA as 
in Algorithm \ref{alg:Q-IRKA} to the quadratized  dynamics  \eqref{eq:final_quad_MU} resulting with a reduced quadratic system with asymptotically stable linear part.
Let $W_T \in \mathbb{R}^{n \times \r}$ denote the first $n$ rows of $W$. Then 
\begin{equation}\label{eq:range_condition}
\text{Range}(W_T)  \subseteq  \text{Range}(\mathcal{C}^T).
\end{equation}
In addition, when $p\le \r$, we have $\text{Range}(W_T)=\text{Range}(\mathcal{C}^T)$.
\end{thm}
{Before we prove Theorem \ref{rankoneW}, we establish a symmetry property of $\Hperm$ that will be used in its proof.}
\begin{lem}\label{lem:Hsymmetric}
Let $\Hperm$ in \eqref{eq:final_quad_MU} be the mode-1 matricization of the third order tensor $\mathcal{H} \in \mathbb{R}^{4n \times 4n \times 4n}$, Then, the tensor $\mathcal{H}$ is symmetric, i.e., for every $\vecv,\vecu \in \mathbb{R}^{4n \times 1}$, it holds
\begin{align}\label{eq:symmetric_prop}
    \Hperm( \vecv \otimes \vecu) = \Hperm (\vecu \otimes \vecv ).
\end{align}
\end{lem}
\begin{proof} Recall \eqref{eq:kronecker_rel} and \eqref{eq:finalHperm}, and rewrite $\Hperm( \vecu \otimes \vecv)$ as
\begin{align}
    \Hperm( \vecu \otimes \vecv)& = \Z P^T( \vecu \otimes \vecv) 
    = \Z ( \vecu\, \mykron\, \vecv), \notag
\end{align}
where 
$
    \vecu = [
    \vecu_1^T \, \vecu_2^T \, \vecu_3^T \, \vecu_4^T
    ]^T $ and $\vecv = [
    \vecv_1^T \, \vecv_2^T \, \vecv_3^T \, \vecv_4^T]^T $ with $ \vecu_i , \vecv_i \in \mathbb{R}^{n \times 1}.
$
Using \eqref{eq:ourKron}  and \eqref{eq:H_tiled} we obtain
{\small
\begin{align}\label{eq:HUkronV}
   \Hperm( \vecu \otimes \vecv) = \begin{bmatrix}
   \Z_2 (\vecu_2 \,\mykron\, \vecv)+ \Z_3(\vecu_3 \,\mykron \, \vecv)+ \Z_4(\vecu_4 \, \mykron\, \vecv)
   \end{bmatrix}.
\end{align}
}
Let $\E_1 = \begin{bmatrix}
     I & 0 & 0 & 0
    \end{bmatrix}\in \mathbb{R}^{n \times 4n}$.
Multiply \eqref{eq:HUkronV}  from the right  with $\E_1$ to pick the first $n$ rows of $\Hperm( \vecu \otimes \vecv)$.
Since the first $n$ rows of $\Z_2, \Z_3$ and $\Z_4$ are zero as can be seen  in \eqref{eq:H2} and \eqref{eq:H}, we have $\E_1\Hperm( \vecu \otimes \vecv) = 0 \in \mathbb{R}^{n \times 1}.$
Similarly, 
$ \E_1\Hperm( \vecv \otimes \vecu)  = 0 \in \mathbb{R}^{n \times 1}.\notag$ Therefore, we have
\begin{align}\label{eq:E1_result}
 \E_1 \Hperm ( \vecu \otimes \vecv) & = \E_1 \Hperm ( \vecv \otimes \vecu).
\end{align}
Now let $\E_2 = \begin{bmatrix}
     0 & I & 0 & 0
    \end{bmatrix}\in \mathbb{R}^{n \times 4n}$
and multiply \eqref{eq:HUkronV}  from the right with $\E_2$:
{\small
\begin{align}\label{eq:E2HUkronV}
 \E_2\Hperm ( \vecu \otimes \vecv) 
 &= \Zthreethree (\vecu_3 \otimes \vecv_3) + \frac{\Zthreefour}{2} (\vecu_3 \otimes \vecv_4)\notag \\ & - \frac{\Zthreefour}{2} (\vecu_4 \otimes \vecv_3) + \Zthreethree (\vecu_4 \otimes \vecv_4). 
\end{align}
}
{Let $S \in \mathbb{R}^{n^2 \times n^2}$ be the permutation matrix} so that
\begin{align} \label{eq:permutation_s}
    \vecu\otimes \vecv = S (\vecv \otimes  \vecu) ,
\end{align}
{where $S$ is partitioned as
   $ S = \begin{bmatrix} \label{eq:permutation_s_decomp}
    S_1 & S_2 & \dots S_n
    \end{bmatrix}$
and each $S_j \in \mathbb{R}^{n^2 \times n}$, for $j=1,\dots,n$, is defined as
\begin{align}\label{eq:S_j}
    S_j = \begin{bmatrix}
    e_j & e_{n+j} & e_{2n+j} & \dots & e_{n(n-1)+j}
    \end{bmatrix},
\end{align}
where $e_j \in \mathbb{R}^{n^2 \times 1}$ denotes 
the $j$th column of a ${n^2 \times n^2}$ identity matrix.
Insert \eqref{eq:permutation_s} into \eqref{eq:E2HUkronV} to obtain
{\small
\begin{align}\label{eq:E2HVKRONU}
    \E_2\Hperm ( \vecu \otimes \vecv)  & = \Zthreethree S (\vecv_3 \otimes \vecu_3) + \frac{\Zthreefour}{2} S (\vecv_3 \otimes \vecu_4) \\
   & ~\quad - \frac{\Zthreefour}{2} S (\vecv_4 \otimes \vecu_3) + \Zthreethree S (\vecv_4 \otimes \vecu_4). \notag
\end{align}
}Then, form $\Zthreethree S$,
\begin{align}\label{eq:PsiS}
    \Zthreethree S = &   \begin{bmatrix}
   \Zthreethree S_1 & \Zthreethree S_2 & \dots \Zthreethree S_n
    \end{bmatrix}, 
\end{align}
where,  using \eqref{eq:S_j} we have, for $j=1,\dots,n$, 
\begin{align}\label{eq:PsiS_j}
    \Zthreethree S_j = \begin{bmatrix}
    \Zthreethree_1(:,j)  & \Zthreethree_2(:,j) & \dots & \Zthreethree_n(:,j) 
    \end{bmatrix},
\end{align}
and  $\Zthreethree_1(:,j)$ is the MATLAB notation referring to the $j$th column of $\Zthreethree_1$.
Since for the coupling between the oscillators $\n$ and $\m$, we have $K_{ij}= K_{ji}$ and  $\gamma_{ij}= \gamma_{ji}$, we can write $\Zthreethree_k(:,j) = \Zthreethree_j(:,k)$ and accordingly rewrite \eqref{eq:PsiS_j} as
\begin{align}\label{eq:PsiS_j_final}
    \Zthreethree S_j  = \Zthreethree_j, \quad  j=1,\dots,n.
\end{align}
Then, inserting \eqref{eq:PsiS_j_final} into \eqref{eq:PsiS} yields
\begin{align}\label{eq:permutationS_effect1}
    \Zthreethree S  = \Zthreethree. 
\end{align}
Similarly, one can show that }
\begin{align}\label{eq:permutationS_effect2}
  \Zthreefour S = - \Zthreefour.
\end{align}
Plug \eqref{eq:permutationS_effect1} and \eqref{eq:permutationS_effect2} into \eqref{eq:E2HVKRONU} to obtain
\begin{align}\label{eq:E2_result}
    \E_2\Hperm ( \vecu \otimes \vecv) = \E_2\Hperm ( \vecv \otimes \vecu).
\end{align}
Following similar steps, one obtains
\begin{align}\label{eq:E4_result}
     \E_k \Hperm ( \vecu \otimes \vecv) & = \E_k\Hperm ( \vecv \otimes \vecu),~~\mbox{for}~k=3,4,
\end{align}
where $\E_3 = \begin{bmatrix}
     0 & 0 & I & 0
    \end{bmatrix}\in \mathbb{R}^{n \times 4n}$ and $\E_4 = \begin{bmatrix}
0 & 0 & 0 & I
\end{bmatrix}\in \mathbb{R}^{n \times 4n}$. Combining \eqref{eq:E1_result}, \eqref{eq:E2_result},  and \eqref{eq:E4_result}, we conclude \eqref{eq:symmetric_prop}, i.e., $\mathcal{H}$ is symmetric.
\end{proof}
{Now we are ready to prove Theorem \ref{rankoneW}.}
\begin{proof}(\emph{of Theorem \ref{rankoneW}})
We start by analyzing the first Sylvester equation involving $W_1$ in Step 7 of Algorithm \ref{alg:Q-IRKA}:
\begin{align}\label{eq:Lyap_W1}
-{E}^{T} W_1 \Lambda - \tilde{A}_\mu^T W_1 = C^T\hat{C} ,   
\end{align}
where $\Lambda = \text{diag}(\lambda_1,\dots,\lambda_r) \in \mathbb{R}^{\r \times \r}$  contains the eigenvalues of $E_r^{-1}A_r$, $\hat{C} \in \mathbb{R}^{\no \times \r}$,
and 
\begin{align} \label{eq:Amu_tild}
\tilde{A}_\mu = \tilde{A} - \mu E = A + H \left((I \otimes \q_0) + (\q_0 \otimes I)\right) - \mu E.
\end{align}
Recall $H \in \mathbb{R}^{4n \times (4n)^2}$ in \eqref{eq:finalHperm}:
\begin{align}
    H  =\begin{bmatrix}
    0_{4n \times 4n^2} & \Z_2P_q^T & \Z_3P_q^T & \Z_4P_q^T
    \end{bmatrix}. \label{eq:H_decom} 
\end{align}
Multiply $\Hperm\left((I \otimes \q_0) + (\q_0 \otimes I)\right)$ 
by the first unit vector $e_1 \in \mathbb{R}^{4n \times 1}$ from the right  to obtain
\begin{align}
    h_1 = \Hperm\left((e_1 \otimes \q_0) + (\q_0 \otimes e_1)\right), \label{eq:h1eq}
\end{align}
where $h_1 \in \mathbb{R}^{4n \times 1}$ denotes the first column of $\Hperm\left((I \otimes \q_0) + (\q_0 \otimes I)\right)$. Now, we apply Lemma \ref{lem:Hsymmetric} to \eqref{eq:h1eq} and use  \eqref{eq:H_decom} to obtain
\begin{align}\label{eq:h_1}
    h_1 = 2 \Hperm (e_1 \otimes \q_0) & = 2 \Hperm \begin{bmatrix}
    q_0 \\ 0_{4n(4n-1)}
    \end{bmatrix} 
    = 0 .
\end{align}
Similarly, for the $k$th column of \eqref{eq:H_decom}, we obtain
\begin{align}\label{eq:h_n}
    h_k = 2 \Hperm (e_k \otimes \q_0) 
   = 0 ~~~ ;  k=2,3,\ldots,n.
\end{align}
Therefore, {\eqref{eq:h_1} and \eqref{eq:h_n}} reveal that the first $n$ columns of 
  $\Hperm\left((I \otimes \q_0) + (\q_0 \otimes I)\right)$ are zero. Using this fact, and \eqref{eq:b_q} and \eqref{eq:a_q},  we obtain the first {$n$} columns of \eqref{eq:Amu_tild} as
  \begin{align}\label{eq:nrowsAmu} 
     \tilde{A}_\mu (:,1:n)  =   \begin{bmatrix}
    -\mu I & 0 & 0 & 0
   \end{bmatrix}^T,
  \end{align}
 {where $Y(:,1:n)$ refers to the first $n$ columns of the matrix $Y$.}
Decompose $  W_1 = \begin{bmatrix}
  W_{11}^T & W_{12}^T & W_{13}^T&W_{14}^T
  \end{bmatrix}^T \in \mathbb{R}^{4n \times \r}$ 
where $W_{1i}$ $\in \mathbb{R}^{n \times \r}~;i=1,2,3,4$.
Recall \eqref{eq:c_q} {and \eqref{eq:b_q}}, and use \eqref{eq:nrowsAmu} in \eqref{eq:Lyap_W1} to obtain 
    $W_{11}(\mu I- \Lambda )  = \mathcal{C}^T \hat{C}$.

Recall that $\mu$ is a positive real number and the reduced order poles are assumed in the left-half plane. Therefore, 
 the matrix $ \mu I - \Lambda $ is invertible and
$
    W_{11} = \mathcal{C}^T\hat{C} (\mu I- \Lambda )^{-1}.
$
Similarly, let $W_2 = [W_{21}^T~ W_{22}^T~  W_{23}^T~ W_{24}^T]^T$ and solve for the first $n \times \r$ block of $W_2$ $\in \mathbb{R}^{4n \times \r}$ denoted by $W_{21} \in \mathbb{R}^{n \times \r}$ in the second Sylvester equation in Step 7 of Algorithm \ref{alg:Q-IRKA}:
\begin{align}\label{eq:W2}
    -{E}^{T} W_2 \Lambda - \tilde{A}_\mu^T W_2 = \mathcal{H}^{(2)} (V_1 \otimes W_1) (\hat{\mathcal{H}}^{(2)})^T.
\end{align}
{Next, we will prove that the first $n$ rows of $\mathcal{H}^{(2)}$ are zero.}  Towards this goal, {recall \eqref{eq:mode1_matricization} and} let $\Hperm$ be {partitioned} as
\begin{align}
   \Hperm =  {\mathcal{\Hperm}^{(1)} } &= \begin{bmatrix}
    \mathcal{\Hperm}_1 & \mathcal{\Hperm}_2 & \dots & \mathcal{\Hperm}_{4n}
    \end{bmatrix},
\end{align}
where $\mathcal{\Hperm}_i \in \mathbb{R}^{4n \times 4n}$ for $i=1,2,\ldots,4n$. From \eqref{eq:H_decom}, we know that the leading $4n^2$ columns are zero. Therefore, we write
\begin{align}
   \Hperm =  {\mathcal{\Hperm}^{(1)} }  = \begin{bmatrix}
    0_{4n \times 4n^2} & \mathcal{\Hperm}_{n+1} & \dots & \mathcal{\Hperm}_{4n}
    \end{bmatrix},
\end{align}
{i.e., 
$\mathcal{\Hperm}_i=0$ for $i=1,2,\ldots,n$.}
 Recall that $\mathcal{H}$ is symmetric due to Lemma \ref{lem:Hsymmetric}. This means that
 \begin{align} \label{H2H3}
     \mathcal{\Hperm}^{(2)} = \mathcal{\Hperm}^{(3)} ~ \in \mathbb{R}^{4n \times (4n)^2},
 \end{align}
where 
$ \mathcal{\Hperm}^{(3)} = \begin{bmatrix}
    \text{vec}(\mathcal{\Hperm}_1) &  \text{vec}(\mathcal{\Hperm}_2) & \dots &  \text{vec}(\mathcal{\Hperm}_{4n})
    \end{bmatrix}^T.$
Since $\mathcal{\Hperm}_i=0$ for $i=1,2,\ldots,n$, we obtain
\begin{align}\label{eq:mathcalH3}
    \mathcal{\Hperm}^{(3)}  = \begin{bmatrix}
 0_{16n^2\times n}  &  \text{vec}(\mathcal{\Hperm}_{n+1}) \dots &  \text{vec}(\mathcal{\Hperm}_{4n})
    \end{bmatrix}^T .
\end{align}
 Due to \eqref{H2H3}, the first $n$ rows of $\mathcal{\Hperm}^{(2)}$ are the first $n$ columns of ${(\mathcal{\Hperm}^{(3)})}^T$. Then, \eqref{eq:mathcalH3} implies that $\mathcal{H}^{(2)}(1:n,:)=0$.
 Using this fact in \eqref{eq:W2} yields
$W_{21} (\mu I -\Lambda)  = 0$.
Since $\mu I -\Lambda$ is invertible, we obtain $W_{21} = 0_{n \times \r}$.
Recall $W$ in Step $8$ of Algorithm \eqref{alg:Q-IRKA} defined as 
$W = W_1 + W_2$. Then, for the first {$n$} rows of $W$, i.e., for $W_T$, we obtain
\begin{align} \label{eq:WTfinal}
     W_T = W_{11} + W_{21} = W_{11} = \mathcal{C}^T\hat{C}(\mu I -\Lambda)^{-1}.
\end{align}
Hence, 
$\text{Range}(W_T) \subseteq   \text{Range}(\mathcal{C}^T)$,
which proves \eqref{eq:range_condition}. In addition,
if $p\le \r$, the matrix $\hat{C}$ has a right-inverse and thus the reverse direction holds as well, i.e., $ \text{Range}(\mathcal{C}^T) \subseteq  \text{Range}(W_T)$, completing the proof of the theorem.
\end{proof}
We emphasize that the structure of the range of $W_T$ due to Q-IRKA as shown in Lemma \ref{rankoneW} is specific to the underlying power network dynamics and its quadratic form. The authors are not aware of any other models that yield such a form. This result will be crucial in forming the proposed algorithm next. 
\subsection{Proposed structure-preserving model reduction algorithm}
Lemma \ref{rankoneW} states that whenever $p\leq \r$ (which will be the common situation), $W_T$ is rank-deficient. Therefore, even though this means that $W_T$ is not suitable to be used as a model reduction space as is,  we can use it to our advantage. 

As noted earlier, we will perform Galerkin projection. In other words, in \eqref{eq:So_ROM_matrices} we will use $\mathcal{W} = \mathcal{V}$. One option is to simply ignore $W_T$ and set $\mathcal{V} = V_T$. However, this is not preferred in an input-output based systems-theoretic model reduction setting since the information in $W_T$ related to the output is fully ignored. Fortunately, due to the subspace result \eqref{eq:range_condition} in Theorem \ref{rankoneW}, we can include both the input-to-state subspace information ($V_T$) and the state-to-output subspace information ($W_T$) 
in one subspace  $\mathcal{V}$ by choosing
in \eqref{eq:So_ROM_matrices} as 
\begin{align}\label{eq:V}
    \mathcal{V} = \text{orth}\left(
    \begin{bmatrix}
    {V}_T & \mathcal{C}^T
    \end{bmatrix}\right) \in \mathbb{R}^{n \times r},
\end{align}
where  ``orth” refers to forming an orthonormal basis so that $\mathcal{V}^T\mathcal{V} = I$. 
Then, we choose $\mathcal{W} = \mathcal{V}$ and perform a Galerkin projection to compute the reduced second-order structure-preserving model \eqref{eq:So_ROM}-\eqref{eq:So_ROM_matrices}.
Thus, despite performing a one-sided Galerkin, by 
exploiting Theorem \ref{rankoneW} and using the reduction basis $\mathcal{V}$ in \eqref{eq:V}, we are able to also incorporate
the output information in the reduction framework, as is the standard in systems theoretic model reduction.
We note that for the column dimension $r$ of $\mathcal{V}$, we have $r = \text{min}(\r+p,2\r)$.

Now, we  finally summarize our proposed approach,
an $\mathcal{H}_2$-based Structure-preserving MOR method for Power Network Dynamics (StrH2), in Algorithm \ref{alg:proposed_meth}. 
We start by converting the original dynamics to the equivalent  quadratic form in Step 1. Then, in Step 2, we 
\emph{either} apply two-sided Q-IRKA (Algorithm \ref{alg:Q-IRKA})
\emph{or} one-sided Q-IRKA (Algorithm \ref{alg:Q-IRKA_onesided}) 
to the resulting quadratic-system to construct $V$. 
The resulting versions of the proposed method will be labeled as 
StrH2-A and StrH2-B, respectively. Step 3 constructs the final model reduction basis using the analysis of Theorem \ref{rankoneW} and Step 4 constructs the final structured reduced-model.

We note that in \eqref{eq:V}, the $W$-subspace information is not needed since it  is already contained in $\mathcal{C}^T$. However, StrH2-A version of Algorithm \ref{alg:proposed_meth}, in Step 2, still uses two-sided Q-IRKA, which computes $W$ during the iteration unlike in one-sided Q-IRKA where the $W$-subspace is never computed. One might think that since the $W$-subspace information is replaced by $\mathcal{C}^T$ in the end, the two options StrH2-A and StrH2-B should produce equivalent results. This is indeed not the case since the $V$-subspace resulting from one-sided and two-sided Q-IRKA implementations are completely different; the former never used the output information. We will see in next section that this distinction does indeed impact the final reduced model.

\begin{algorithm}
   \caption{$\mathcal{H}_2$-based Structure-preserving MOR for Power Network Dynamics (StrH2) }\label{alg:proposed_meth}
    \begin{algorithmic}[1]
    \Statex \textbf{Input}: Second-order model \eqref{eq:So_Dyn}-\eqref{eq:So_Out} 
    \begin{align*}
             \mathcal{M} \ddot{\delta}(t) + \mathcal{D} \dot{\delta}(t) + f(\delta) = \mathcal{B} u(t)
    ~,~ y(t) = \mathcal{C}\delta(t)
    \end{align*}
    \Statex \textbf{Output}: Reduced second-order  model \eqref{eq:So_ROM}-\eqref{eq:So_ROMOut}
    \begin{align*}
             \mathcal{M}_r \ddot{\delta_r}(t) + \mathcal{D}_r \dot{\delta_r}(t) + f_r(\delta_r) =  \mathcal{B}_r u(t) ~ , ~ y_r(t)  =  \mathcal{C}_r\delta_r(t)
    \end{align*}
    \State Transform the second-order dynamic \eqref{eq:So_Dyn}-\eqref{eq:So_Out} into quadratic dynamics \eqref{eq:final_quad_MU} with asymptotically stable linear part:
    \begin{align} \tag{\ref{eq:final_quad_MU}}
\begin{array}{rcl}
 {E}\dot \z(t) & \hspace{-1ex} = \hspace{-1ex} & \tilde{A_{\mu}}\z(t) + \Hperm(\z(t) \otimes \z(t)) + \tilde{B} \tilde{u}(t) \\ 
 y(t) & \hspace{-1ex} = \hspace{-1ex}  & C\z(t), ~~~ \z(0) = 0.
 \end{array}
\end{align}
    \State  For the quadratic dynamics \eqref{eq:final_quad_MU}, 
    apply \vspace{1.5ex}
    \Statex \hspace{0.15in} Option A: Algorithm \ref{alg:Q-IRKA}  
    \Statex ~\hspace{0.65in} or \hspace{0.85in} to construct $V$.
    \Statex \hspace{0.15in} Option B: Algorithm \ref{alg:Q-IRKA_onesided} \vspace{1.5ex}
    \State Form the final model reduction basis $\mathcal{V}$ using the leading $n$ rows of $V$ and the output matrix $\mathcal{C}$ as in \eqref{eq:V}:
    \begin{align*}
       \mathcal{V} = \text{orth}\left( \begin{bmatrix}
    {V}_T & \mathcal{C}^T
    \end{bmatrix}\right)
    \end{align*}
    \State Compute the reduced matrices as in \eqref{eq:So_ROM_matrices} with $\mathcal{W}=\mathcal{V} $.
    \end{algorithmic}
    \end{algorithm}

\section{Numerical Experiments}\label{sec:numerical_result}
We test the proposed approach on  two models: (i) the SM model of the 39-bus New England test system and (ii) IEEE 118 bus system,  included in the MATPOWER software toolbox  \cite{ZimmermanMurillo2016},\cite{ZimmermanThomas2010}. We focus on a single-output system $(\no = 1)$ 
by choosing the arithmetic mean of all phase angles as the output as 
 in \cite{safaeeGugercin22021,TobiasGrundel2020}.
Models are generated by MATPOWER and MATLAB toolbox \href{ https://sourceforge.net/projects/pg-sync-models/ }{ \texttt{pg-sync-models}}
 \cite{NishikawaMotter2015} that provide all the necessary parameters to form \eqref{eq:So_Dyn} including $J_i$, $D_i$, $\omega_R$, $K_{ij}$, $\gamma_{ij}$ and $B_i$.  
  To illustrate the effectiveness of our approach, we compare it with two other approaches (i) a structure-preserving formulation of balanced truncation for power systems denoted by {Str-QBT} and (ii) POD. 
 
 For Str-QBT, we follow the approach of \cite{TobiasGrundel2020} in finding  $V$ and $W$. However, as opposed to \cite{TobiasGrundel2020}, which performs the reduction on the quadratic dynamics \eqref{eq:final_quad_MU} and constructs a quadratic reduced model, we directly find the reduced second order dynamic \eqref{eq:So_ROM} by extracting $\mathcal{V}$ and $\mathcal{W}$ from $V$ and $W$.
 To be more precise,  \cite{TobiasGrundel2020} forms the reduction bases ${V}$ as
  \begin{align} \label{eq:BT_quadraticbases}
      V & = \text{blkdiag}\{{V_b},{V_b},{V_b},{V_b}\}
 \end{align}
 and similarly for $W$ using $W_b$
where  ${V_b}$ and ${W_b} \in \mathbb{R}^{n \times \r}$ are 
 \begin{align} \label{eq:BT_bases}
 \begin{array}{rcl}
      {V_b} & \hspace{-1ex} = \hspace{-1ex} & R_b^T \hat{U}(:,1:\r)(\hat{\Sigma}(:,1:\r))^{-\frac{1}{2}} \\
      {W_b} & \hspace{-1ex} = \hspace{-1ex} & S_b^T \hat{V}(:,1:\r)(\hat{\Sigma}(:,1:\r))^{-\frac{1}{2}},
      \end{array}
 \end{align}
 $R_b S_b^T = \hat{U} \hat{\Sigma} \hat{V}^T$ is the SVD of $R_b S_b^T$, and $R_b$ and $S_b$ are, respectively, the Cholesky factors of the second $n \times n$ block of the truncated reachability and observability gramians of the quadratic system \cite{BennerGoyal2017}, \cite{TobiasGrundel2020}. Then, \cite{TobiasGrundel2020} performs the reduction on \eqref{eq:final_quad_MU} using \eqref{eq:BT_quadraticbases} to obtain a reduced quadratic model for \eqref{eq:final_quad_MU}.
 {Instead, since the second $n \times n$ blocks of the truncated reachability and observability gramians of the quadratic system correspond to the second order dynamic \eqref{eq:So_Dyn},} we pick $\mathcal{V} = V_b$ and $\mathcal{W} = W_b$ in \eqref{eq:So_ROM_matrices} and directly perform reduction on the second order dynamic \eqref{eq:So_Dyn}-\eqref{eq:So_Out} to obtain \eqref{eq:So_ROM}, \eqref{eq:So_ROMOut}.
 
To apply POD, we simulate \eqref{eq:So_Dyn} for a given time interval to form state snapshot matrix $\Delta \in \mathbb{R}^{n \times L}$ and compute its SVD: 
\begin{align*}
    \Delta  = \begin{bmatrix}
    \delta(t_0) & \delta(t_1) & \dots & \delta(t_{L-1})
    \end{bmatrix} = U_{\Delta} \Sigma_{\Delta} V_{\Delta}^T.
\end{align*}
Then, {the $r$ leading left singular vectors of $\Delta$}, i.e., the leading $r$ columns of 
$U_\Delta$,
form the reduction basis $\mathcal{V} = \mathcal{W}$ in \eqref{eq:So_ROM_matrices}. 

To test the accuracy of the reduced models, we will measure the time-domain error between the true output 
$y(t)$ and the reduced-model outputs $y_r(t)$ over a time interval $t\in [0,T]$. Towards this goal, 
define the $\mathcal{L}_{\infty}(T)$ norm of  $y(t)$ as
\begin{equation}  \label{ytLinf}
\| y \|_{\mathcal{L}_{\infty}(T)} = \max_{t \in [0,T]} \mid y(t) \mid.~~
\end{equation}
 The corresponding \emph{relative} $\mathcal{L}_{\infty}$ output error 
 is defined as 
\begin{equation}   \label{etLinf}
 \|e \|_{\mathcal{L}_{\infty}(T)} = \frac{\|y - y_r\|_{\mathcal{L}_{\infty}(T)}}{\|y\|_{\mathcal{L}_{\infty}(T)}}.~~
 \end{equation}
 \subsection{New England test system}
The  model has order $n=39$ in the original second-order coordinates~\eqref{eq:So_Dyn}. We choose  $\mu = 10^{-3}$ in \eqref{eq:Atildemu}. 
We apply both formulations of {StrH2 in Algorithm \ref{alg:proposed_meth} (StrH2-A and StrH2-B), Str-QBT}, and POD; and reduce the order to $r = 2,\dots\,25$. 
In Figure \ref{fig:comparison_u_1}, the relative $\mathcal{L}_{\infty}$ error  
$\|e \|_{\mathcal{L}_{\infty}(T)}$ with $T=10$ vs. reduction order $r$
is depicted for each method. {The figure shows that for the majority of $r$ values, {StrH2-A} yields the lowest relative $\mathcal{L}_\infty$ error.} We note that this is the best case scenario for POD since the reduced model is tested with the same input that is used to train POD. 
\graphicspath{{figures/}} 
\begin{figure}[!t]
\centering
  \includegraphics[width=2.5in]{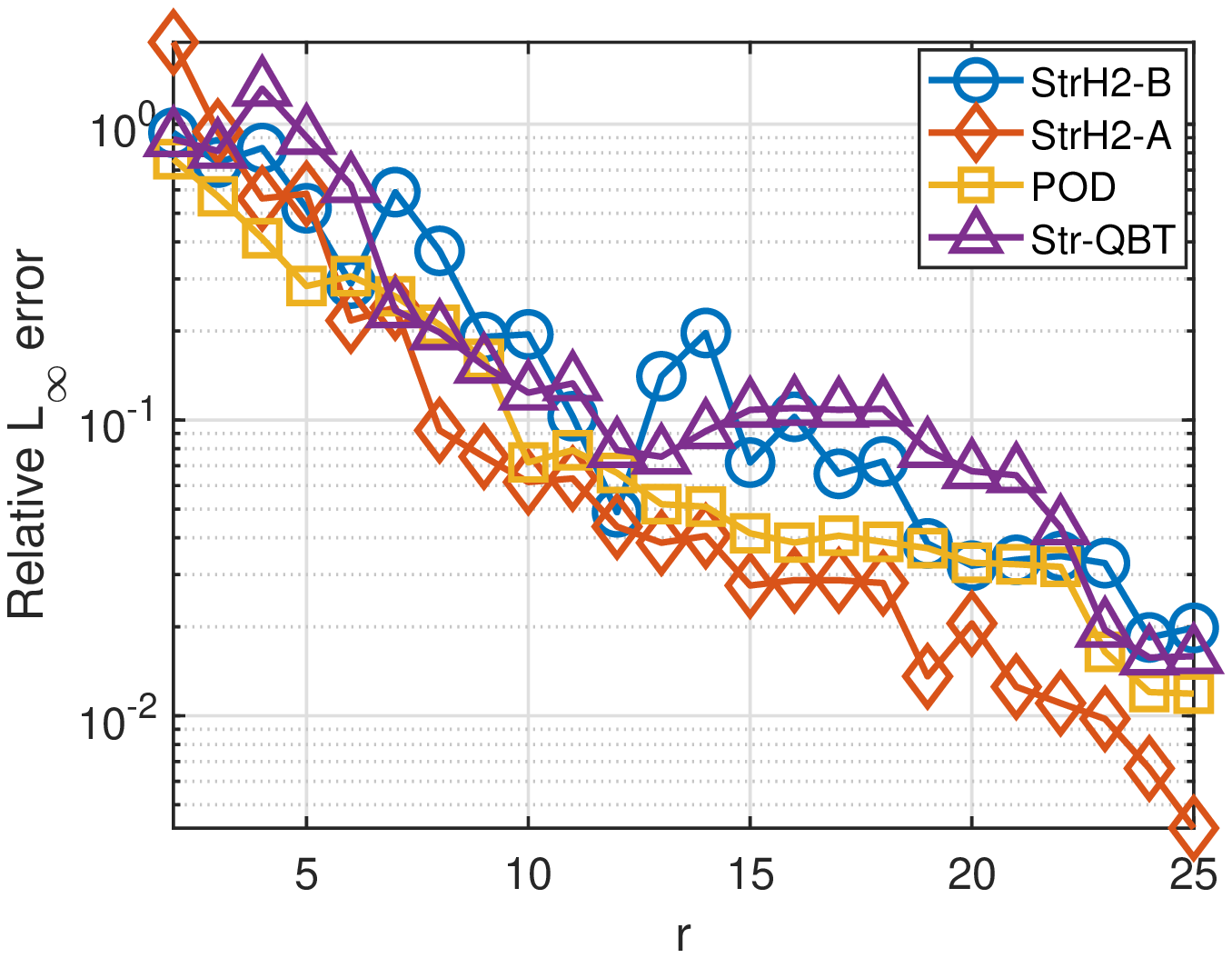}
  \caption{Relative $\mathcal{L}_{\infty}$ error vs. reduction order}
  \label{fig:comparison_u_1}
  \vspace{-1em}
\end{figure}

For a more detailed comparison, in Figure \ref{fig:error_output_u_1b}, the 
 output of the original and all the reduced order models of order $r=23$ (at which the relative error due to {StrH2-A} drops below $1\%$), and the corresponding absolute error, $\mid y(t)-y_r(t)\mid$, are depicted. 
  This figure supports the earlier observation that  {StrH2-A} approximates the original model with a better accuracy compared to  POD and {Str-QBT}.
\graphicspath{{figures/}} 
\begin{figure}[!t]
\centering
  \includegraphics[width=2.5in]{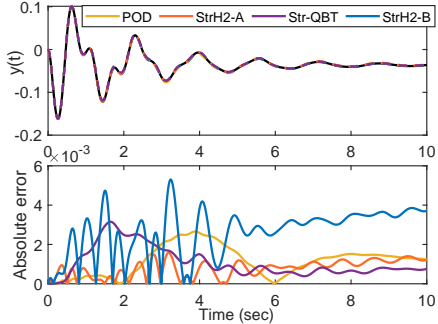}
  \caption{Original and the reduced order models output and the corresponding absolute error vs. time for $r=23$}
  \label{fig:error_output_u_1b}
  \vspace{-1em}
\end{figure}
To investigate the robustness of the algorithms to  variations in the input, we slightly perturb the input and repeat the same model reduction procedures as before. One might view this perturbation to the input as if the operation conditions are slightly varied. We 
perturb the input only by $0.1\%$.
Note that the choice of input has never entered the proposed algorithm or {Str-QBT}; however it is used to train the POD model. Figures \ref{fig:comparison_u_1.001} shows the relative error as $r$ varies for this slightly varied input. As in the previous case, {StrH2-A} outperforms the other methods. We also observe that the relative errors for the POD reduced models stagnate after $r=13$, i.e., not much further improvements, while {StrH2-A and StrH2-B}, and {Str-QBT}  approximate the original model with almost the same  accuracy as in the earlier case. This is expected since the input never entered the model reduction process of the 
{StrH2-A, StrH2-B or Str-QBT}. Thus, for this \emph{small} input variation case, those reduced models provide accurate approximations independent of the input selection as in the linear case.  
\graphicspath{{figures/}} 
\begin{figure}[!t]
\centering
  \includegraphics[width=2.5 in]{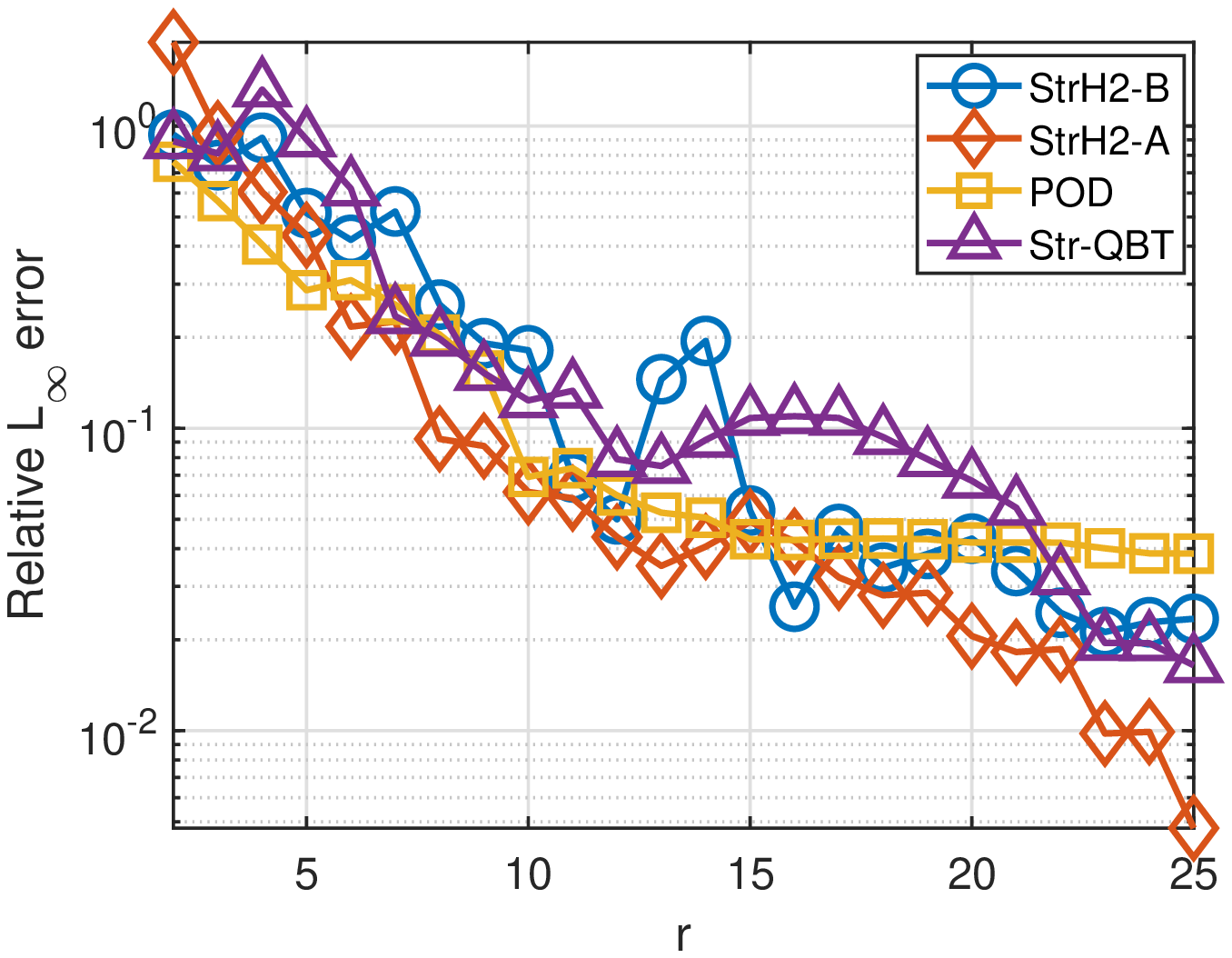}
  \caption{Relative $\mathcal{L}_{\infty}$ error vs. reduction order}
  \label{fig:comparison_u_1.001}
  \vspace{-1em}
\end{figure}

\subsection{IEEE 118 bus system}
The second-order model \eqref{eq:So_Dyn} in this case has $n = 118$. We choose {$\mu = 10^{-2}$} and  apply {StrH2-A, StrH2-B} and POD for $r = 2,\dots\,10$. We use $T=[0~3]$. {Str-QBT} has been skipped  due to its poor performance in this example. The relative $\mathcal{L}_{\infty}$ output error \eqref{etLinf} vs. $r$ is depicted in Figure \ref{fig:RelError_Case118_u1}. Although POD yields a lower $\mathcal{L}_{\infty}$ error (the testing input is the same as the training input for POD), {StrH2-A} has an acceptable accuracy as POD.
\graphicspath{{figures/}} 
\begin{figure}
     \centering
     \begin{subfigure}[b]{0.5\textwidth}
         \centering
         \includegraphics[width=2.5in]{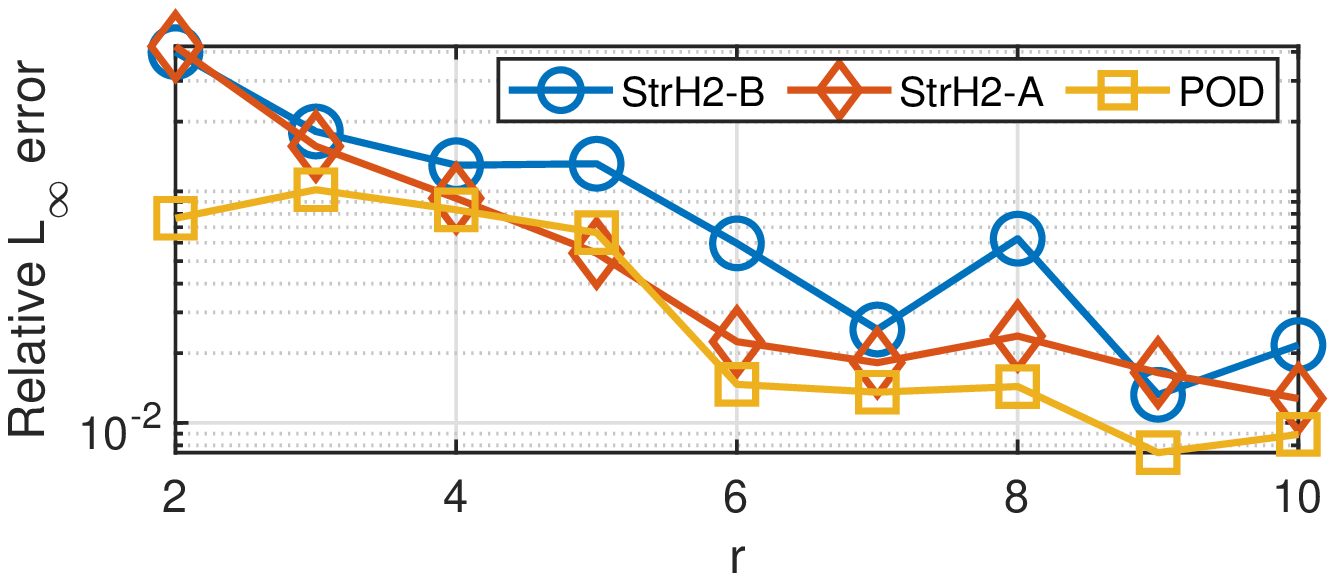}
         \vspace{-.8em}
         \caption{$u$ is not perturbed}
         \label{fig:RelError_Case118_u1}
     \end{subfigure}
     \vfill
     \begin{subfigure}[b]{0.5\textwidth}
         \centering
         \includegraphics[width=2.5in]{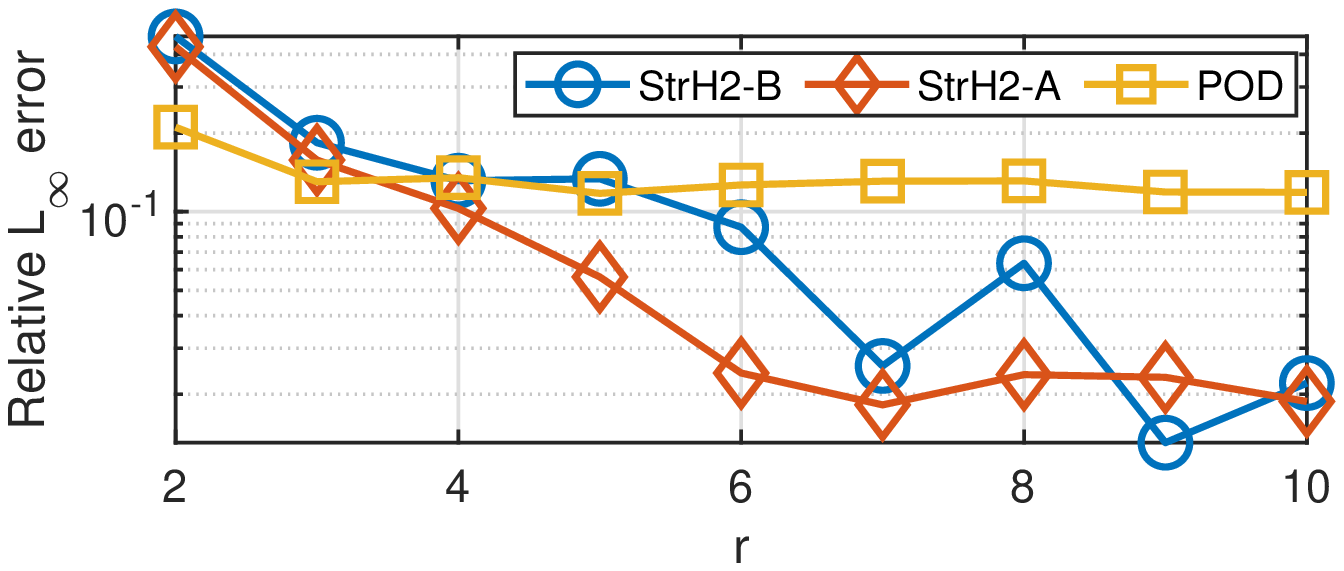}
         \vspace{-.8em}
       \caption{$u$ is perturbed}
       \label{fig:RelError_Case118_u1.005}
     \end{subfigure}
     \caption{Relative $\mathcal{L}_{\infty}$ error vs. reduction order}
       \vspace{-1em}
\end{figure}
As in the previous example, to test the robustness of the algorithms in the presence of a variation in the input, we slightly perturb the input $u$ by $0.5\%$. As one can observe in Figure \ref{fig:RelError_Case118_u1.005}, both {StrH2-A and StrH2-B} outperform POD. Also as in the previous example, the relative errors for the POD reduced models  do not show further improvements after $r=3$, while both {StrH2} formulations approximate the original model with similar accuracy as in the former case.

\section{Conclusion}\label{sec:conclusion}
{We have developed a structure-preserving system-theoretic MOR technique for nonlinear power grid networks. We lifted the original nonlinear model to its equivalent quadratic form in order to employ Q-IRKA as an intermediate stage of our MOR framework. We have shown that the model reduction bases obtained by Q-IRKA have a specific subspace structure that can be exploited to construct the desired reduction basis for reducing the original nonlinear model. This reduction basis has led to a reduced model that preserved the physically meaningful second-order structure of the original model. We have illustrated the effectiveness of our proposed approach via two numerical examples. 
}

\bibliographystyle{plain}
\bibliography{SafaeeGugercin}

\begin{thebibliography}{10}

\bibitem{AnnakkageMartinez2012}
U.~D. Annakkage, N.~K.~C. Nair, Y.~Liang, A.~M. Gole, V.~Dinavahi,
  B.~Gustavsen, T.~Noda, H.~Ghasemi, A.~Monti, M.~Matar, R.~Iravani, and J.~A.
  Martinez.
\newblock Dynamic system equivalents: A survey of available techniques.
\newblock {\em IEEE Transactions on Power Delivery}, 27(1):411--420, 2012.

\bibitem{AntBG20}
A.~C. Antoulas, C.~Beattie, and S.~G\"{u}\u{g}ercin.
\newblock {\em Interpolatory methods for model reduction}.
\newblock Computational Science and Engineering 21. SIAM, Philadelphia, 2020.

\bibitem{BennerBreiten2012_2}
P.~Benner and T.~Breiten.
\newblock Interpolation-based $\mathcal{H}_2$-model reduction of bilinear
  control systems.
\newblock {\em SIAM Journal on Matrix Analysis and Applications},
  33(3):859--885, 2012.

\bibitem{BennerTobias2015}
P.~Benner and T.~Breiten.
\newblock Two-sided projection methods for nonlinear model order reduction.
\newblock {\em SIAM Journal on Scientific Computing}, 37(2):B239--B260, 2015.

\bibitem{BennerGoyal2017}
P.~Benner and P.~Goyal.
\newblock Balanced truncation model order reduction for quadratic-bilinear
  control systems.
\newblock e-print 1705.00160, arXiv, 2017.

\bibitem{BennerGoyalGugercin2018}
P.~Benner, P.~Goyal, and S.~Gugercin.
\newblock $\mathcal{H}_2$-quasi-optimal model order reduction for
  quadratic-bilinear control systems.
\newblock {\em SIAM Journal on Matrix Analysis and Applications},
  39(2):983--1032, 2018.

\bibitem{BESSELINKSchilders2013}
B.~Besselink, U.~Tabak, A.~Lutowska, N.~{van de Wouw}, H.~Nijmeijer, D.J.
  Rixen, M.E. Hochstenbach, and W.H.A. Schilders.
\newblock A comparison of model reduction techniques from structural dynamics,
  numerical mathematics and systems and control.
\newblock {\em Journal of Sound and Vibration}, 332(19):4403--4422, 2013.

\bibitem{ChaniotisPai2005}
D.~{Chaniotis} and M.~A. {Pai}.
\newblock Model reduction in power systems using {Krylov} subspace methods.
\newblock {\em IEEE Transactions on Power Systems}, 20(2):888--894, 2005.

\bibitem{ChaturantabutSorensen2010}
S.~Chaturantabut and D.~C. Sorensen.
\newblock Nonlinear model reduction via discrete empirical interpolation.
\newblock {\em SIAM Journal on Scientific Computing}, 32(5):2737--2764, 2010.

\bibitem{ChengJ2018}
X.~Cheng and J.~M.~A. Scherpen.
\newblock Clustering approach to model order reduction of power networks with
  distributed controllers.
\newblock {\em Advances in Computational Mathematics}, 44(6):1917--1939, Dec
  2018.

\bibitem{CheridBettayeb1991}
A.~Cherid and M.~Bettayeb.
\newblock Reduced-order models for the dynamics of a single-machine power
  system via balancing.
\newblock {\em Electric Power Systems Research}, 22(1):7 -- 12, 1991.

\bibitem{Chow2013}
J.~H. Chow.
\newblock {\em Power system coherency and model reduction}, volume~84.
\newblock Springer, 2013.

\bibitem{ChowSauer1990}
J.~H. Chow, J.~R. Winkelman, M.~A. Pai, and P.~W. Sauer.
\newblock Singular perturbation analysis of large-scale power systems.
\newblock {\em International Journal of Electrical Power \& Energy Systems},
  12(2):117 -- 126, 1990.

\bibitem{Gu2011}
C.~{Gu}.
\newblock Qlmor: A projection-based nonlinear model order reduction approach
  using quadratic-linear representation of nonlinear systems.
\newblock {\em IEEE Transactions on Computer-Aided Design of Integrated
  Circuits and Systems}, 30(9):1307--1320, 2011.

\bibitem{HinzeVolkwein2005}
M.~Hinze and S.~Volkwein.
\newblock Proper orthogonal decomposition surrogate models for nonlinear
  dynamical systems: Error estimates and suboptimal control.
\newblock In P.~Benner, D.~C. Sorensen, and V.~Mehrmann, editors, {\em
  Dimension Reduction of Large-Scale Systems}, pages 261--306. Springer Berlin
  Heidelberg, 2005.

\bibitem{IshizakiAihara2015}
T.~Ishizaki, K.~Kashima, A.~Girard, J.~Imura, L.~Chen, and K.~Aihara.
\newblock Clustered model reduction of positive directed networks.
\newblock {\em Automatica}, 59:238 -- 247, 2015.

\bibitem{Kolda2009}
T.~G. Kolda and B.~W. Bader.
\newblock Tensor decompositions and applications.
\newblock {\em SIAM Review}, 51(3):455--500, 2009.

\bibitem{LeungVillella2019}
J.~{Leung}, M.~{Kinnaert}, J.~{Maun}, and F.~{Villella}.
\newblock Model reduction of coherent {LPV} models in power systems.
\newblock In {\em 2019 IEEE Power Energy Society General Meeting (PESGM)},
  pages 1--5, 2019.

\bibitem{LevronBelikov2017}
Y.~{Levron} and J.~{Belikov}.
\newblock Reduction of power system dynamic models using sparse
  representations.
\newblock {\em IEEE Transactions on Power Systems}, 32(5):3893--3900, 2017.

\bibitem{Liu2009}
S.~Liu.
\newblock {\em Dynamic-data driven real-time identification for electric power
  systems}.
\newblock PhD thesis, University of Illinois at Urbana-Champaign, 2009.

\bibitem{MalikDiez2016}
M.~H. Malik, D.~Borzacchiello, F.~Chinesta, and P.~Diez.
\newblock Reduced order modeling for transient simulation of power systems
  using trajectory piece-wise linear approximation.
\newblock {\em Advanced Modeling and Simulation in Engineering Sciences},
  3(1):31, Dec 2016.

\bibitem{mccormick1976computability}
G.~P. McCormick.
\newblock Computability of global solutions to factorable nonconvex programs:
  Part i—convex underestimating problems.
\newblock {\em Mathematical programming}, 10(1):147--175, 1976.

\bibitem{mlinaric2015efficient}
P.~Mlinari{\'c}, S.~Grundel, and P.~Benner.
\newblock Efficient model order reduction for multi-agent systems using {QR}
  decomposition-based clustering.
\newblock In {\em Proceedings of 54th IEEE Conference on Decision and Control},
  pages 4794--4799, 2015.

\bibitem{MlinaricIshizaki2018}
P.~{Mlinarić}, T.~{Ishizaki}, A.~{Chakrabortty}, S.~{Grundel}, P.~{Benner},
  and J.~{Imura}.
\newblock Synchronization and aggregation of nonlinear power systems with
  consideration of bus network structures.
\newblock In {\em 2018 European Control Conference (ECC)}, pages 2266--2271,
  2018.

\bibitem{NishikawaMotter2015}
T.~Nishikawa and A.~E Motter.
\newblock Comparative analysis of existing models for power-grid
  synchronization.
\newblock {\em New Journal of Physics}, 17(1):015012, 2015.

\bibitem{OsipovSun2018}
D.~{Osipov} and K.~{Sun}.
\newblock Adaptive nonlinear model reduction for fast power system simulation.
\newblock {\em IEEE Transactions on Power Systems}, 33(6):6746--6754, 2018.

\bibitem{PaiAdgaonkar1981}
M.~A. {Pai} and R.~P. {Adgaonkar}.
\newblock Singular perturbation analysis of nonlinear transients in power
  systems.
\newblock In {\em 1981 20th IEEE Conference on Decision and Control including
  the Symposium on Adaptive Processes}, pages 221--222, 1981.

\bibitem{ParriloMarsden1999}
P.~A. {Parrilo}, S.~{Lall}, F.~{Paganini}, G.~C. {Verghese}, B.~C. {Lesieutre},
  and J.~E. {Marsden}.
\newblock Model reduction for analysis of cascading failures in power systems.
\newblock In {\em Proceedings of the 1999 American Control Conference (Cat. No.
  99CH36251)}, volume~6, pages 4208--4212, 1999.

\bibitem{PurvineWang2017}
E.~Purvine, E.~Cotilla-Sanchez, M.~Halappanavar, Z.~Huang, G.~Lin, S.~Lu, and
  S.~Wang.
\newblock Comparative study of clustering techniques for real-time dynamic
  model reduction.
\newblock {\em Statistical Analysis and Data Mining: The ASA Data Science
  Journal}, 10(5):263--276, 2017.

\bibitem{QiDimitrovski2017}
J.~{Qi}, J.~{Wang}, H.~{Liu}, and A.~D. {Dimitrovski}.
\newblock Nonlinear model reduction in power systems by balancing of empirical
  controllability and observability covariances.
\newblock {\em IEEE Transactions on Power Systems}, 32(1):114--126, 2017.

\bibitem{QianWillcox2020}
E.~Qian, B.~Kramer, B.~Peherstorfer, and K.~Willcox.
\newblock Lift \& learn: Physics-informed machine learning for large-scale
  nonlinear dynamical systems.
\newblock {\em Physica D: Nonlinear Phenomena}, 406:132401, 2020.

\bibitem{RamirezKamalasadan2016}
A.~{Ramirez}, A.~{Mehrizi-Sani}, D.~{Hussein}, M.~{Matar}, M.~{Abdel-Rahman},
  J.~{Jesus Chavez}, A.~{Davoudi}, and S.~{Kamalasadan}.
\newblock Application of balanced realizations for model-order reduction of
  dynamic power system equivalents.
\newblock {\em IEEE Transactions on Power Delivery}, 31(5):2304--2312, 2016.

\bibitem{TobiasGrundel2020}
T.~K.S. Ritschel, F.~Weiß, M.~Baumann, and S.~Grundel.
\newblock Nonlinear model reduction of dynamical power grid models using
  quadratization and balanced truncation.
\newblock {\em at - Automatisierungstechnik}, 68(12):1022--1034, 2020.

\bibitem{safaeeGugercin22021}
B.~Safaee and S.~Gugercin.
\newblock Data-driven modeling of power networks.
\newblock In {\em 2021 60th IEEE Conference on Decision and Control (CDC)},
  pages 4236--4241, 2021.

\bibitem{safaeeGugercin2021}
B.~Safaee and S.~Gugercin.
\newblock Structure-preserving model reduction of parametric power networks.
\newblock In {\em 2021 American Control Conference (ACC)}, pages 1824--1829,
  2021.

\bibitem{SturkSandberg2014}
C.~{Sturk}, L.~{Vanfretti}, Y.~{Chompoobutrgool}, and H.~{Sandberg}.
\newblock Coherency-independent structured model reduction of power systems.
\newblock {\em IEEE Transactions on Power Systems}, 29(5):2418--2426, 2014.

\bibitem{Schaft2014}
A.~van~der Schaft.
\newblock On model reduction of physical network systems.
\newblock In {\em Proceedings of 21st International Symposium on Mathematical
  Theory of Networks and Systems (MTNS)}, pages 1419--1425, 2014.

\bibitem{WangYu2013}
C.~{Wang}, H.~{Yu}, P.~{Li}, C.~{Ding}, C.~{Sun}, X.~{Guo}, F.~{Zhang},
  Y.~{Zhou}, and Z.~{Yu}.
\newblock Krylov subspace based model reduction method for transient simulation
  of active distribution grid.
\newblock In {\em 2013 IEEE Power Energy Society General Meeting}, pages 1--5,
  2013.

\bibitem{WangDing2015}
C.~Wang, H.~Yu, P.~Li, J.~Wu, and C.~Ding.
\newblock Model order reduction for transient simulation of active distribution
  networks.
\newblock {\em IET Generation, Transmission \& Distribution}, 9:457--467(10),
  April 2015.

\bibitem{WangZhou2012}
S.~{Wang}, S.~{Lu}, G.~{Lin}, and N.~{Zhou}.
\newblock Measurement-based coherency identification and aggregation for power
  systems.
\newblock In {\em 2012 IEEE Power and Energy Society General Meeting}, pages
  1--7, 2012.

\bibitem{WangPai2014}
S.~{Wang}, S.~{Lu}, N.~{Zhou}, G.~{Lin}, M.~{Elizondo}, and M.~A. {Pai}.
\newblock Dynamic-feature extraction, attribution, and reconstruction ({DEAR})
  method for power system model reduction.
\newblock {\em IEEE Transactions on Power Systems}, 29(5):2049--2059, 2014.

\bibitem{ZhaoRen2017}
H.~Zhao, X.~Lan, and H.~Ren.
\newblock Nonlinear power system model reduction based on empirical gramians.
\newblock {\em Journal of Electrical Engineering}, 68(6):425 -- 434, 2017.

\bibitem{ZhaoShi2013}
H.~S. Zhao, N.~Xue, and N.~Shi.
\newblock Nonlinear {Dynamic} {Power} {System} {Model} {Reduction} {Analysis}
  {Using} {Balanced} {Empirical} {Gramian}.
\newblock {\em Applied Mechanics and Materials}, 448-453:2368--2374, October
  2013.

\bibitem{ZimmermanMurillo2016}
R.~D. Zimmerman and C.~E. Murillo-S{\'a}nchez.
\newblock Matpower 6.0 user’s manual.
\newblock {\em Power Systems Engineering Research Center}, 9, 2016.

\bibitem{ZimmermanThomas2010}
R.~D. Zimmerman, C.~E. Murillo-S{\'a}nchez, and R.~J. Thomas.
\newblock Matpower: Steady-state operations, planning, and analysis tools for
  power systems research and education.
\newblock {\em IEEE Transactions on Power Systems}, 26(1):12--19, 2010.

\end{thebibliography}

\end{document}